\documentclass[preprint,12pt]{elsarticle}

\usepackage{amssymb}
\usepackage{amsmath,amsfonts}
\usepackage{algorithmic}
\usepackage{array}
\usepackage[caption=false,font=normalsize,labelfont=sf,textfont=sf]{subfig}
\usepackage{textcomp}
\usepackage{stfloats}
\usepackage{url}
\usepackage{verbatim}
\usepackage{graphicx}
\def\BibTeX{{\rm B\kern-.05em{\sc i\kern-.025em b}\kern-.08em
    T\kern-.1667em\lower.7ex\hbox{E}\kern-.125emX}}
\usepackage{balance}

\usepackage{booktabs}                        
\usepackage{subcaption}   
\usepackage{color}
\usepackage{amsmath}
\usepackage{amssymb}
\usepackage{graphicx}
 \usepackage{dsfont}
\usepackage{amsfonts}
\usepackage{amsthm} 
\usepackage{stmaryrd}
\usepackage[all]{xy}
\usepackage{multirow}

\usepackage{qcircuit}
 
\def\>{\ensuremath{\rangle}}
\def\<{\ensuremath{\langle}}

\newcommand{\hs}{\mathcal{H}}

\newtheorem{thm}{Theorem}[section]

\newtheorem{lem}{Lemma}[section]
\newtheorem{defn}{Definition}[section]

\newtheorem{exam}{Example}[section]

\newtheorem{ques}{Question}[section]

\journal{}

\begin{document}

\begin{frontmatter}

\title{Verification of Recursively Defined Quantum Circuits}

\author{Mingsheng Ying and Zhicheng Zhang}

\affiliation{organization={Centre for Quantum Software and Information, University of Technology Sydney}, 
           addressline={15 Broadway}, 
           city={Ultimo},
           postcode={2007}, 
          state={NSW},
            country={Australia}}

\begin{abstract}
Recursive techniques have recently been introduced into quantum programming so that a variety of large quantum circuits and algorithms can be elegantly and economically programmed. In this paper, we present a proof system for formal verification of the correctness of recursively defined quantum circuits. The soundness and (relative) completeness of the proof system are established. To demonstrating its effectiveness, a series of application examples of the proof system are given, including (multi-qubit)  controlled gates, a quantum circuit generating (multi-qubit) GHZ (Greenberger-Horne-Zeilinger) states, recursive definition of quantum Fourier transform, quantum state preparation, and quantum random-access memories (QRAM). 
\end{abstract}


\end{frontmatter}

\section{Introduction}

The procedure and recursion constructs have been widely used in classical programming for modularity. In particular, an invocation of a (recursive) procedure enables a (unbounded) number of uses of a uniquely defined function or program section. In order to support the same level of abstraction, several recursive techniques started to emerge in quantum programming. For example, recursive quantum programs  with \textit{classical control flows} (e.g. defined by measurement outcomes) were first introduced in \cite{Selinger}. Their semantics was further studied in Section 3.4 of \cite{Ying16}. Furthermore, a proof system was developed in \cite{Xu21} for verification of this class of recursive quantum programs. A particularly interesting work is \cite{Deng24} where a class of recursive quantum (unitary) programs (i.e. \textit{recursively defined quantum circuits}) were defined and a synthesis framework of them using SMT solvers was proposed. 
A proof technique for verifying this class of recursive quantum (unitary) programs was also developed in \cite{Deng24} based on the path-sum method for verification of quantum circuits \cite{Amy18}. 

In this paper, we introduce a novel notion of quantum recursive procedure than \cite{Deng24}. The key design ideas of this quantum recursion scheme include: \begin{enumerate}\item \textit{quantum if-statements} \cite{Alt05, Ying12, Sab18, Yuan, Bich, Voi} are employed for defining quantum control flows; \item recursive procedures are \textit{parameterised by classical variables} (and expressions). 
\end{enumerate} To the best of the authors' knowledge, this kind of quantum recursion has not been considered in the current literature. The importance of this new scheme of quantum recursion is demonstrated by a series of examples that show that using it, large quantum gates, circuits and algorithms can be recursively defined in a compact and elegant way. In particular, a quantum Divide-and-Conquer strategy (not classical Divide-and-Conquer strategy used in the design of quantum algorithms) can be elegantly programmed in this language (see Examples in Section \ref{sec-examples} and Appendix \ref{appendix-examples}). 

\subsection{Motivating Example: Recursive Definition of Quantum Fourier Transform}

To appreciate the power of our quantum recursive definition, let us consider quantum Fourier transform (QFT) as the first example. The QFT has been used as a key subroutine in many important quantum algorithms, including the celebrated Shor factoring algorithm. The QFT on $n$ qubits is mathematically defined by 
\begin{equation}\label{QFT}\mathit{QFT}^{(n)}|j\rangle=\frac{1}{\sqrt{2^n}}\sum_{k=0}^{2^n-1} e^{2\pi ijk/2^n}|k\rangle
\end{equation} for $j=0,1,...,2^n -1.$ If we use the binary representation $j=j_1j_2...j_n=\sum_{l=1}^n j_l 2^{n-l}$ and binary fraction $0.k_1k_2...k_m=\sum_{l=1}^mk_l 2^{-l},$ then the defining equation (\ref{QFT}) of $\mathit{QFT}^{(n)}$ can be rewritten as
\begin{equation}\label{QFT-1}\mathit{QFT}^{(n)}|j_1,...,j_n\rangle=\frac{1}{\sqrt{2^n}}\bigotimes_{l=1}^n\left(|0\rangle+e^{2\pi i 0.j_{n-l+1}...j_n}|1\rangle\right).
\end{equation} As shown in the standard textbook \cite{NC00}, $\mathit{QFT}^{(n)}$ can be decomposed into the sequence of single-qubit and two-qubit basic gates in Table \ref{QFT-table}, where $H$ is the Hadamard gate, and for each $l\in\{1,...,n\}$, $C(R_l)$ is a controlled-rotation with the following single-qubit gate as the rotation operator: $$R_l=\left(\begin{array}{cc}1&0\\ 0&e^{2\pi i/2^l}\end{array}\right).$$
\begin{table*}[t]
\begin{align*}\mathit{QFT}^{(n)}[q_1,...,q_n]::=\ &H[q_1]; C(R_2)[q_2,q_1];C(R_3)[q_3,q_1];...;C(R_{n-1})[q_{n-1},q_1];C(R_n)[q_n,q_1];\\
&H[q_2];C(R_2)[q_3,q_2];C(R_3)[q_4,q_2];...;C(R_{n-1})[q_n,q_2];\\
&\qquad\qquad\qquad\qquad ..............................\\
&H[q_{n-1}];C(R_2)[q_n,q_{n-1}];\\
&H[q_n]; \\
&\mathit{Reverse}[q_1,...,q_n]
\end{align*}
\caption{A quantum circuit for quantum Fourier transform. Here,  $C(R_l)[q_i,q_j]$ stands for the controlled-$R_l$ with $q_i$ as its control qubit and $q_j$ as its target qubit, and $\mathit{Reverse}[q_1,...,q_n]$ is the quantum gate that reverses the order of qubits $q_1,...,q_n$.}\label{QFT-table}
\end{table*}

In a quantum programming language that does not support recursion, one has to program $\mathit{QFT}^{(n)}$ in a way similar to Table \ref{QFT-table}. For a large number $n$ of qubits, the size of such a $\mathit{QFT}^{(n)}$ program will be very large and unmanageable. Using the  recursion scheme to be  defined in this paper,  however, we can write $\mathit{QFT}^{(n)}$ as a program, of which the size is independent of the number $n$ of qubits:   

\begin{exam}[Recursive Definition of QFT]\label{example-QFT} Our recursive definition of QFT is inspired by the Cooley-Tukey algorithm \cite{Cooley65} for FFT (Fast Fourier Transform). Following the idea of \cite{Cooley65}, we can rearrnge the right-hand side of equation (\ref{QFT}) into two parts: a sum over the even-numbered basis states and a sum over the odd-numbered basis states; that is, for $k=0,1,...,2^{n-1}-1$, 
\begin{equation}\label{QFT0}|k\rangle\rightarrow\sum_{v=0}^{2^{n-1}-1}e^{\frac{2\pi ivk}{2^{n-1}}} |2v\rangle +e^{\frac{2\pi ik}{2^n}}\sum_{v=0}^{2^{n-1}-1}e^{\frac{2\pi ivk}{2^{n-1}}}|2v+1\rangle,\end{equation}
\begin{equation}\label{QFT1}|k+2^{n-1}\rangle\rightarrow \sum_{v=0}^{2^{n-1}-1}  e^{\frac{2\pi ivk}{2^{n-1}}}|2v\rangle -e^{\frac{2\pi ik}{2^n}}\sum_{v=0}^{2^{n-1}-1}e^{\frac{2\pi ivk}{2^{n-1}}}|2v+1\rangle.\end{equation}
Now let $k,v\in\{0,1,...,2^{n-1}-1\}$ be thought of as $(n-1)$-bit strings. Then when representing them as $n$-bit strings, we have: $|k\rangle=|0k\rangle$, $|k+2^{n-1}\rangle=|1k\rangle$, $|2v\rangle=|v0\rangle$ and $|2v+1\rangle=|v1\rangle$. 
Thus, equations (\ref{QFT0}) and (\ref{QFT1}) can be rewritten as 
\begin{equation}\mathit{QFT}^{(n)}|0k\rangle=\mathit{QFT}^{(n-1)}|k\rangle\otimes \frac{1}{\sqrt{2}}\left(|0\rangle+e^{\frac{2\pi ik}{2^n}}|1\rangle\right),\end{equation}
\begin{equation}\mathit{QFT}^{(n)}|1k\rangle=\mathit{QFT}^{(n-1)}|k\rangle\otimes \frac{1}{\sqrt{2}}\left(|0\rangle-e^{\frac{2\pi ik}{2^n}}|1\rangle\right).\end{equation}
We introduce the single-qubit quantum gate:
\begin{equation}
    \label{eq:gate_S}
    S_{\theta}=\frac{1}{\sqrt{2}}\left(\begin{array}{cc}1 & 1 \\ e^{\pi i\theta} & - e^{\pi i\theta}\end{array}\right).
\end{equation}
Let $q[m:n]$ stand for the array of $n-m+1$ qubits indexed by integers $[m:n]=\{m,m+1,...,n\}$. Then QFT on $q[m:n]$ can be recursively programmed as: 
\begin{equation}
\label{p-QFT-1}
    \begin{split}
        \mathit{QFT}(m,n)\Leftarrow\ &\mathbf{if}\ m=n\ \mathbf{then}\ S_{0}[q[m]]\\
        &\qquad\quad\ \ \ \mathbf{else}\ \mathit{Rot}(m,n,0);\mathit{QFT}(m+1,n);\mathit{Shift}(m,n)\\ 
        &\mathbf{fi},
    \end{split}
\end{equation}
\begin{equation}
    \label{p-QFT-m}
    \begin{split}
        \mathit{Rot}(m,n,\theta)\Leftarrow\ &\mathbf{if}\ m=n\ \mathbf{then}\ S_{\theta}[q[m]]\\
        &\qquad\quad\ \ \ \mathbf{else}\ \mathbf{qif}[q[n]]\left(\square_{x=0}^1 |x\rangle \rightarrow \mathit{Rot}(m,n-1,(\theta+x)/2)\right)\ \mathbf{fiq}\\
        &\mathbf{fi},
    \end{split}
\end{equation}
\begin{equation}
    \label{p-QFT-2}
    \begin{split}
        \mathit{Shift}(m,n)\Leftarrow\ &\mathbf{if}\ m< n\ \mathbf{then}\\
        &\qquad \mathit{Swap}[q[m],q[n]];\mathit{Shift}(m+1,n)\\
        &\mathbf{fi}.
    \end{split}
\end{equation}
\end{exam}

The above recursive definition of QFT is much more economic than the quantum circuits in Table \ref{QFT-table}. But it is not easy to see that the recursive programs (\ref{p-QFT-1}), (\ref{p-QFT-m}) and (\ref{p-QFT-2}) actually realise the transform (\ref{QFT}). This observation leads us to the following:  


\begin{ques}\label{ques-1} \textbf{How to verify the correctness of the recursive definition $\mathit{QFT}(m,n)$ of QFT?} \end{ques}

It seems that the correctness of $\mathit{QFT}(m,n)$ cannot be verified by the proof technique developed in \cite{Deng24} and any other existing technique, since quantum if-statement ($\mathbf{qif}$) is employed in equation (\ref{p-QFT-m}), and no verification  techniques for programs with $\mathbf{qif}$ have been proposed in the current literature.   

\subsection{Contributions of the Paper}\label{Contri}

The aim of this paper is to extend the verification method of classical recursive programs (see for example Chapters 4 and 5 of \cite{Apt09} and Chapter 6 of \cite{Francez}) for verifying the correctness of recursively defined quantum circuits. As an application, we answer Question \ref{ques-1} in Section \ref{sec-examples}. 

The key ideas of our approach and the main contributions of the paper are described as follows:\begin{itemize}\item \textbf{Correctness specification}: Since in our programming language, classical variables are introduced as parameters in defining quantum circuits, we choose to use the following form:
\begin{equation}\label{Hoare-00}\{A,|\varphi\rangle\}\ C\ \{B,|\psi\rangle\}\end{equation} of Hoare triple for specifying the correctness of a quantum circuit $C$, where $A$ and $B$ are classical first-order logical formulas, and $|\varphi\rangle,|\psi\rangle$ are two quantum states (parameterized by classical variables). Intuitively, Hoare triple (\ref{Hoare-00}) means: 
\begin{itemize}\item \textit{whenever we start in a classical state $\sigma$ that satisfies precondition $A$, and a quantum state $|\varphi\rangle$ is inputted to the circuit $C$, then the execution of $C$ will terminates in a classical state $\sigma^\prime$ that satisfies postcondition $B$, and the circuit outputs quantum state $|\psi\rangle$.}\end{itemize}

It should be particularly pointed out that classical logical predicates $A,B$ and quantum states $|\varphi\rangle, |\psi\rangle$ in Hoare triple (\ref{Hoare-00}) are related to each other, because the classical parameters in $|\varphi\rangle$ and $|\psi\rangle$ are given by classical states $\sigma$ and $\sigma^\prime$, respectively (see Subsection \ref{sec-correctness} for a detail explanation). 


\item \textbf{Proof System}: To reason about the correctness of recursively defined quantum circuits, we develop a set of proof rules defined in terms of Hoare triples of the form (\ref{Hoare-00}). Among these rules, several are new and have no counterparts in verification of classical programs; for example, a rule is devised for reasoning about quantum if-statement, and a rule is introduced for preserving the linearity of quantum circuits. However, others are essentially obtained by tailoring their classical counterparts so that they are suited to proving correctness formulas of the form (\ref{Hoare-00}) specially designed for quantum circuits with classical parameters. In particular, the rule for reasoning about recursion is essentially the standard Scott induction rule adapted to our circumstance of recursively defined quantum circuits. There are many different formulations of the induction rule in the literature on verification of classical recursive programs. We generalise the one presented in Chapters 4 and 5 of textbook \cite{Apt09}. 
\textit{The soundness and (relative) completeness of our proof system are established.}  
\end{itemize}

\subsection{Organisation of the Paper} For convenience of the reader, in Section \ref{sec-language}, we introduce the syntax and semantics of a class of quantum recursive programs, namely recursively defined quantum circuits. 
Hoare triples of the form (\ref{Hoare-00}) are formally defined in Section \ref{sec-correctness} for specifying correctness of quantum recursive programs. Our proof system is then presented in Section \ref{proof-system}. The soundness and (relative) completeness theorems of the proof system are stated in Section \ref{sec-completeness}, but their involved proofs are deferred to Appendix \ref{proof-sound} for a better readability. As an example, the correctness proof of recursive QFT is given in Section \ref{sec-examples} to illustrate how our proof system can actually be applied in verification of recursive quantum circuits and algorithms. More examples are presented in Appendix \ref{appendix-examples}, including controlled gates (by multiple qubits), generation of GHZ (Greenberger-Horne-Zeilinger) states, quantum state preparation, and quantum random access memories (QRAM). 

\section{A Language for Recursive Definition of Quantum Circuits}\label{sec-language}

Let us start from a programming language language for recursive definition of quantum circuits  and its semantics.

\subsection{Quantum Arrays}\label{sec-arrays}

This quantum programming language employs a simple quantum data type, namely quantum array. So, let us formally define it in this subsection.  

\subsubsection{Quantum types} A basic quantum type $\hs$ denotes an intended Hilbert space. A higher quantum type is defined of the form: 
\begin{equation}\label{h-type}T_1\times ...\times T_n\rightarrow\hs,\end{equation}
where $T_1,...,T_n$ are basic classical types, and $\hs$ is a basic quantum type. It denotes the tensor product of multiple copies of $\hs$ indexed by $v_1\in T_1,...,v_n\in T_n$: 
\begin{equation}\label{h-space}\hs^{\otimes (T_1\times ...\times T_n)}=\bigotimes_{v_1\in T_1,...,v_n\in T_n}\hs_{v_1,...,v_n},\end{equation} where $\hs_{v_1,...,v_n}=\hs$ for all $v_1\in T_1,...,v_n\in T_n$. 
For example, let $\hs_2$ be the $2$-dimensional Hilbert space as the state space of a qubit. Then a qubit array $q$ can be defined of type $\mathbf{integer}\rightarrow\hs_2$, where $\mathbf{integer}$ is the (classical) integer type. For any two integers $k\leq l$, section $q[k:l]$ stands for the restriction of $q$ to the interval $[k:l]=\{{\rm integer}\ i\mid k\leq i\leq l\}$. Note that (\ref{h-space}) may an infinite tensor product of Hilbert spaces \cite{vN39} when one of $T_1,...,T_n$ is infinite (e.g., the type $\mathbf{integer}\rightarrow\hs_2$). However, all application examples in Section \ref{sec-examples} and Appendix \ref{appendix-examples} only need finite tensor products (e.g., the section $q[k:l]$). 
 
\subsubsection{Quantum variables} In our programming language, we use two sorts of quantum variables:\begin{itemize}\item simple quantum variables, each of a basic quantum type, say $\hs$; \item array quantum variables, each of a higher quantum type, say $T_1\times ...\times T_n\rightarrow\hs$.  
\end{itemize}

For simplicity of presentation, we often consider a simple variable of type $\hs$ as an array variable of type $T_1\times ...\times T_n\rightarrow\hs$ with $n=0$. Let $q$ be an array quantum variable of the type $T_1\times ...\times T_n\rightarrow\hs$. For any values $v_1\in T_1,...,v_n\in T_n$, array element $q[v_1,...,v_n]$ can be seen as a quantum variable of type $\hs$. Moreover, we introduce the notion of subscripted quantum variable:  

\begin{defn}\label{def-sub-variable} For each $1\leq i\leq n$, let $s_i$ be a classical expression of type $T_i$. Then $q[s_1,...,s_n]$ is called a subscripted quantum variable of type $\hs$.\end{defn}

For example, if $q:\mathbf{integer}\times\mathbf{integer}\rightarrow\hs_2$ is a two-dimensional array of qubits, and $x,y$ are integer variables, then $q[2x-y, 3x+4y-75]$ is a subscripted qubit variable. 

\subsection{A Quantum Circuit Description Language}\label{sec-circuits}

Now we are ready to define our language for quantum recursive programming. It can be described as an expansion of a simple quantum circuit description language. In this subsection, let us first formally define the quantum circuit description language $\mathbf{QC}$.     

\subsubsection{Syntax}We assume that the alphabet of \textbf{QC} consists of: \begin{itemize}\item A set $\mathit{QV}$ of simple quantum variables (but no subscripted quantum variables, because no classical variables are involved in $\mathbf{QC}$);
\item A set $\mathcal{U}$ of elementary unitary gates.   
\end{itemize}
The elementary unitary gates in $\mathcal{U}$ will be instantiated in practical applications. For example, $\mathcal{U}$ can be chosen as the set $U(2)$ of all single qubit gates, or $\mathcal{U}=\{H,T\}$, where $T$ stands for the $\pi/8$ gate.  
We fix the following notations:
\begin{itemize}\item[--] Each quantum variable $q\in\mathit{QV}$ assumes a type $T(q)$, meaning that $q$ stands for a quantum system with the Hilbert space denoted by $T(q)$ as its state space.
\item[--] A sequence $\overline{q}=q_1,...,q_n$ of distinct quantum variables, called a quantum register, denotes a composite quantum system consisting of subsystems $q_1,...,q_n$. Its type is defined as the tensor product $T(\overline{q})=T(q_1)\otimes ...\otimes T(q_n)$. 
\item[--] Each elementary unitary gate $U\in\mathcal{U}$ assumes a type of the form $T(U)=\hs_1\otimes ...\otimes\hs_n$. If $\overline{q}=q_1,...,q_n$ is a quantum register with $T(q_i)=\hs_i$ for $i=1,...,n$, then $U[\overline{q}]$ is an eligible quantum gate application. In this case, we say that the types of $U$ and $\overline{q}$ match. 
\end{itemize}

Upon the alphabet $\mathit{QV}$ and $\mathcal{U}$, we can formally describe the syntax of quantum circuits: 

\begin{defn}\label{def-circ}Quantum circuits $C\in\mathbf{QC}$ are defined by the syntax:
\begin{equation}\label{def-syntax}C::=\ U[\overline{q}]\mid C_1;C_2 \mid \mathbf{qif}[q] \left(|0\rangle\rightarrow C_0\right)\square \left(|1\rangle\rightarrow C_1\right) \mathbf{fiq}\end{equation}
More precisely, they are inductively defined by the following clauses, where $\mathit{qv}(C)$ denotes the quantum variables in $C$:
\begin{enumerate} 
\item[(1)] \textbf{Elementary gates}: If $U\in\mathcal{U}$ is a elementary unitary gate and $\overline{q}$ is a quantum register such that their types match, then quantum gate $U[\overline{q}]$ is a circuit, and $\mathit{qv}(U[\overline{q}])=\overline{q}$;  
\item[(2)] \textbf{Sequential composition}: If $C_1$ and $C_2$ are circuits, then $C\equiv C_1;C_2$ is a circuit too, and $\mathit{qv}(C)=\mathit{qv}(C_1)\cup\mathit{qv}(C_2);$
\item[(3)] \textbf{Quantum if-statement}: If $q$ is a qubit, and $C_0,C_1$ are circuits with $q\notin \mathit{qv}(C_0)\cup\mathit{qv}(C_1)$, then \begin{equation}\label{eq-if}C\equiv \mathbf{qif}[q] \left(|0\rangle\rightarrow C_0\right)\square\left(|1\rangle\rightarrow C_1\right) \mathbf{fiq}\end{equation} is a circuit, and $\mathit{qv}(C)=\{q\}\cup\mathit{qv}(C_0)\cup\mathit{qv}(C_1).$
\end{enumerate}\end{defn}

Intuitively, each $C\in\mathbf{QC}$ represents a circuit with each quantum variable in $\mathit{qv}(C)$ as a quantum wire. The quantum if-statement (\ref{eq-if}) is a quantum counterpart of the if-statement in classical programming languages. The qubit variable $q$ in (\ref{eq-if}) is often called a quantum coin. Intuitively, If quantum coin $q$ is in the state $|0\rangle$, then program $C$ behaves like $C_0$, and if $q$ is in the state $|1\rangle$, then $C$ behaves like $C_1$. The fundamental difference between quantum and classical if-statements is that quantum coin $q$ can be in a superposition of $|0\rangle$ and $|1\rangle$, and thus it defines a quantum control flows. This will be seen shortly when the semantics of (\ref{eq-if}) is defined.
We need to point out emphatically that the condition $q\notin \mathit{qv}(C_0)\cup\mathit{qv}(C_1)$ is necessary for the quantum if-statement~(\ref{eq-if}) to be meaningful in physics.  
 This condition explicitly requires that quantum coin $q$ must be external to the quantum circuits $C_0$ and $C_1$ controlled by it.  
 For more detailed discussion, we refer to \cite{Alt05, Ying12, Sab18, Yuan, Bich, Voi} as well as Sections 1.1.2 and 6.1 of \cite{Ying16}. The quantum if-statement is the key that enables our programming language to support quantum recursion (with quantum control flow), as will be seen later (in particular, through the examples in Section \ref{sec-examples} and Appendix \ref{appendix-examples}).  

\subsubsection{Operational semantics} To define the semantics of language $\mathbf{QC}$, for each quantum variable $q\in\mathit{QV}$, we assume that its type $T(q)$ denotes Hilbert space $\hs_q$. Then for any set $X\subseteq\mathit{QV}$ of quantum variables, the state space of the quantum variables in $X$ is the tensor product    
$\hs_X=\bigotimes_{q\in X}\hs_q.$ In particular, the state space of all quantum variables is $\hs_\mathit{QV}$. A configuration over $X$ is defined as a pair $(C,|\psi\rangle)$, where $C\in\mathbf{QC}$ is a quantum circuit with $\mathit{qv}(C)\subseteq X$ or $C=\ \downarrow$ (standing for termination), and $|\psi\rangle$ is a pure quantum state in Hilbert space $\hs_X$. We write $\mathcal{C}(X)$ for the set of all configurations over $X$. The operational semantics of quantum circuits in $\mathbf{QC}$ is defined in terms of transitions between configurations of the form:
\begin{equation}\label{eq-transition}(C,|\varphi\rangle)\rightarrow (C^\prime,|\psi\rangle).\end{equation} The transition (\ref{eq-transition}) means that if the quantum variables are currently in state $|\varphi\rangle$, then after we execute one step of quantum circuit $C$ on them, they will be in state $|\psi\rangle$, with the circuit remaining to execute being $C^\prime$. Formally, we have: 

\begin{defn}The operational semantics of circuits in $\mathbf{QC}$ with quantum variables in $X$ is the transition relation $\rightarrow\ \subseteq\mathcal{C}(X)\times\mathcal{C}(X)$ between configurations that is defined by the transitional rules given in Table \ref{circuit-rules}.
\begin{table*}[t]
\begin{equation*}\begin{split}&
({\rm GA})\ \ 
(U[\overline{q}],|\psi\rangle)\rightarrow \left(
\downarrow,(U\otimes I_{X\setminus\overline{q}})|\psi\rangle\right)\qquad\qquad\qquad
({\rm SC})\ \ \frac{(C_1,|\psi\rangle)\rightarrow (
C_1^{\prime},|\psi^{\prime}\rangle)}{(
C_1;C_2,|\psi\rangle)\rightarrow (C_1^{\prime};C_2,|\psi^\prime\rangle)}\\
&({\rm QIF})\ \ \frac{|\psi\rangle=\alpha_0|0\rangle_{q}|\theta_0\rangle+\alpha_1|1\rangle_{q}|\theta_1\rangle\qquad\qquad (C_i,|\theta_i\rangle)\rightarrow^\ast(\downarrow,|\theta_i^\prime\rangle)\ (i=0,1)}{\left(\mathbf{qif}[q] \left(|0\rangle\rightarrow C_0\right)\square\left(|1\rangle\rightarrow C_1\right) \mathbf{fiq},|\psi\rangle\right)\rightarrow \left(\downarrow,\alpha_0|0\rangle_{q}|\theta_0^\prime\rangle+\alpha_1|1\rangle_{q}|\theta_1^\prime\rangle\right)}
\end{split}\end{equation*}
\caption{Transition Rules for Quantum Circuits. In rule (GA), $I_{X\setminus\overline{q}}$ stands for the identity operator on space $\hs_{X\setminus\overline{q}}$. In rule (SC), we assume $\downarrow;C_2=C_2.$ In rule (QIF), $\rightarrow^\ast$ denotes the reflexive and transitive closure of relation $\rightarrow$.}\label{circuit-rules}
\end{table*}
\end{defn}

The transition rules (GA) and (SC) are self-explanatory. But the rule (QIF) needs an careful explanation. First, for any state $|\psi\rangle$ in a Hilbert space $\hs_X$ with $q\in X$, we can always write it in the form of $|\psi\rangle=\alpha_0|0\rangle_{q}|\theta_0\rangle+\alpha_1|1\rangle_{q}|\theta_1\rangle$ with $|\theta_0\rangle,|\theta_1\rangle\in \hs_{X\setminus\{q\}}$, as done in the premise of the rule (QIF). Second, 
the condition $q\notin\mathit{qv}(C_0)\cup\mathit{qv}(C_1)$ in the clause (3) of Definition \ref{def-circ} warrants that $\mathit{qv}(C_0),\mathit{qv}(C_1)\subseteq X\setminus\{q\}$. Thus, $(C_0,|\theta_0\rangle),(C_1,|\theta_1\rangle)$ are eligible configurations.  
Third, let us write $\rightarrow^n$ for the composition of $n$ copies of $\rightarrow$. Then it is possible that $(C_0,|\theta_0\rangle)\rightarrow^{n_0}(\downarrow,|\theta_0^\prime\rangle)$ and $(C_1,|\theta_1\rangle)\rightarrow^{n_1}(\downarrow,|\theta_1^\prime\rangle)$ for $n_0\neq n_1$. So, the reflexive and transitive closure $\rightarrow^\ast$ of $\rightarrow$ is used in the premise of the rule (QIF). Finally, let $C\in\mathbf{QC}$ be a quantum circuit and $\mathit{qv}(C)\subseteq X$. Then the denotational semantics of $C$ over $X$ can be defined as operator $\llbracket C\rrbracket$ on $\hs_X$ as follows: $$\llbracket C\rrbracket|\varphi\rangle=|\psi\rangle\ {\rm if\ and\ only\ if}\ (C,|\varphi\rangle)\rightarrow^\ast(\downarrow,|\psi\rangle).$$ In particular, it is easy to show that the denotational semantics of quantum if-statement:  
\begin{equation}\begin{split}\label{qcase-3}\left\llbracket\mathbf{qif}[q] \left(|0\rangle\rightarrow C_0\right)\square\left(|1\rangle\rightarrow C_1\right) \mathbf{fiq}\right\rrbracket &=|0\rangle\langle 0|\otimes\llbracket C_0\rrbracket+|1\rangle\langle 1|\otimes\llbracket C_1\rrbracket =\left(\begin{array}{cc} \llbracket C_0\rrbracket & 0\\ 0 &\llbracket C_1\rrbracket\end{array}\right).
\end{split}\end{equation} 

For example, the ${\rm CNOT}$ gate with $q_1$ as the control qubit and $q_2$ as the target qubit can be written as $\mathbf{qif}[q_1] \left(|0\rangle\rightarrow I[q_2]\right)\square \left(|1\rangle\rightarrow X[q_2]\right)$, where $I$ and $X$ are the identity and ${\rm NOT}$ gates, respectively. Furthermore, if $\mathcal{U}=U(2)$, then $\mathbf{QC}$ is universal in the sense that every finite-dimensional unitary operator can be represented by a circuit in $\mathbf{QC}$, and if $\mathcal{U}=\{H,T\}$, then $\mathbf{QC}$ is approximately universal; that is, every finite-dimensional unitary operator can be approximated by a sequence of circuits in $\mathbf{QC}$.      

\subsection{Quantum Circuits Defined with Classical Variables}\label{sec-c-variables}

To enable quantum circuits be defined using classical expressions as their parameters, we assume a fixed classical programming language, and embed $\mathbf{QC}$ defined in the previous subsection into it. 
 \textit{It is always supposed that this classical language is deterministic.} The obtained quantum language is then denoted $\mathbf{QC}^+$.

\subsubsection{Syntax} Formally, the alphabet of $\mathbf{QC}^+$ consists of the alphabet of the chosen classical language, which we choose not to specify explicitly here for simplicity, together with the alphabet of $\mathbf{QC}$ (expanded with quantum array variables). Then we have:  
\begin{defn}\label{def-circ-1}Quantum circuits $C\in\mathbf{QC^+}$ with classical variables are defined by the syntax:
\begin{equation}\label{def-syntax-1}\begin{split}C::=\ \mathbf{skip} &\mid \overline{x}:=\overline{t}\mid U[\overline{q}]\ |\ C_1;C_2\\ &\mid \mathbf{if}\ b\ \mathbf{then}\ C_1\ \mathbf{else}\ C_2\ \mathbf{fi}\mid \mathbf{qif}[q] \left(|0\rangle\rightarrow C_0\right)\square\left(|1\rangle\rightarrow C_1\right) \mathbf{fiq}
\end{split}\end{equation}
where $\overline{x}$ is a string of classical (simple or subscripted) variables, $\overline{t}$ is a string of classical expressions, $b$ is a Boolean expression, $q$ and $\overline{q}$ can be simple or subscripted quantum variables, and other conditions are the same as in Definition \ref{def-circ}.

The quantum variables $\mathit{qv}(C)$ in $C\in\mathbf{QC}^+$ are defined inductively as follows:\begin{enumerate} 
\item $\mathit{qv}(\mathbf{skip})=\mathit{qv}(\overline{x}:=\overline{t})=\emptyset$ and $\mathit{qv}(U[\overline{q}])=\overline{q}$;
\item $\mathit{qv}(C_1;C_2)=\mathit{qv}(\mathbf{if}\ b\ \mathbf{then}\ C_1\ \mathbf{else}\ C_2\ \mathbf{fi})=\mathit{qv}(C_1)\cup\mathit{qv}(C_2)$;
\item $\mathit{qv}(\mathbf{qif}[q] \left(|0\rangle\rightarrow C_0\right)\square\left(|1\rangle\rightarrow C_1\right) \mathbf{fiq})=\{q\}\cup\mathit{qv}(C_0)\cup \mathit{qv}(C_1).$
\end{enumerate} \end{defn}

As usual in classical programming, a conditional of the form $\mathbf{if}\ b\ \mathbf{then}\ C_1\ \mathbf{else}\ \mathbf{skip}\ \mathbf{fi}$ will be simply written as $\mathbf{if}\ b\ \mathbf{then}\ C_1\ \mathbf{fi}$. 

\subsubsection{Operational semantics}\label{sec-sem-q++}
To define operational semantics of $\mathbf{QC}^+$, we need to introduce several new notations for handling subscripted quantum variables. For a classical state $\sigma$ and a classical expression $t$, we write $\sigma(t)$ for the value of $t$ in $\sigma$. If $q$ is a simple quantum variable, then $\sigma(q)=q$, and if $q$ is a subscripted quantum variable $q^{\prime}[s_1,...,s_n]$, then $\sigma(q)$ denotes the quantum system $q^{\prime}[\sigma(s_1),...,\sigma(s_n)]$. For a sequence $\overline{q}=q_1...q_k$ of simple or subscripted quantum variable, $\sigma(\overline{q})$ denotes the composite system $\sigma(q_1)...\sigma(q_k)$. If $q_i=q_i^\prime[s_{i1},...,s_{in_i}]$ for $1\leq i\leq k$, then we write $\mathit{Dist}(\overline{q})$ for the following first-order logical formula:
\begin{equation}\label{def-dist}(\forall i,j\leq k)\left(i\neq j\wedge q^\prime_i=q^\prime_j\rightarrow (\exists l\leq n_i)\left(s_{il}\neq s_{jl}\right)\right).\end{equation}
Thus, $\sigma\models\mathit{Dist}(\overline{q})$ means that in the classical state $\sigma$, $\overline{q}$ denotes $k$ distinct quantum systems. We also need to modify the definition of configuration in order to accommodate classical variables. Let $X$ be a set of simple or array quantum variables. As in the previous subsection, we write $\hs_X=\bigotimes_{q\in X}\hs_q$. But here $q$ may be a quantum array variables; in this case, $\hs_q$ is defined by (\ref{h-space}). 
A configuration over $X$ is now defined as a triple $(C,\sigma,|\psi\rangle)$, where $C\in\mathbf{QC}^+$ is a circuit  with $\mathit{qv}(C)\subseteq X$, $\sigma$ is a state of classical variables, and $|\psi\rangle$ is a pure quantum state in $\hs_X$. For simplicity, we abuse a bit of notation and still use $\mathcal{C}(X)$ to denote the set of all configurations over $X$. Accordingly, a transition between configurations appears as 
\begin{equation}\label{eq-transition-1}(C,\sigma, |\varphi\rangle)\rightarrow (C^\prime,\sigma^\prime, |\psi\rangle), \end{equation} meaning that if the current state of the classical variables is $\sigma$, and the quantum variables are currently in state $|\varphi\rangle$, then after executing one step of quantum circuit $C$, the states of classical and quantum variables become $\sigma^\prime$ and $|\psi\rangle$, respectively, with the circuit remaining to execute being $C^\prime$. Then, we have: 

\begin{defn}\label{def-sem-classical}The operational semantics of circuits in $\mathbf{QC}^+$ with quantum variables in $X$ is the transition relation $\rightarrow\ \subseteq\mathcal{C}(X)\times\mathcal{C}(X)$ between configurations defined by the transitional rules given in Table \ref{circuit-rules-1}.
\begin{table*}[t]
\begin{equation*}\begin{split}&({\rm SK})\ \ 
(\mathbf{skip},\sigma, |\psi\rangle)\rightarrow (\downarrow,\sigma, |\psi\rangle)\qquad\qquad\qquad\quad ({\rm AS})\ \ (\overline{x}:=\overline{t},\sigma, |\psi\rangle)\rightarrow (\downarrow, \sigma[\overline{x}:=\sigma(\overline{t})], |\psi\rangle)\\ 
&({\rm GA})\ \ 
\frac{\sigma\models\mathit{Dist}(\overline{q})}{(U[\overline{q}], \sigma, |\psi\rangle)\rightarrow (
\downarrow,\sigma, U_{\sigma(\overline{q})}|\psi\rangle)}\qquad\qquad ({\rm SC})\ \ \frac{(C_1,\sigma, |\psi\rangle)\rightarrow (
C_1^{\prime},\sigma^\prime, |\psi^{\prime}\rangle)}{(
C_1;C_2,\sigma, |\psi\rangle)\rightarrow (C_1^{\prime};C_2,\sigma^\prime, |\psi^\prime\rangle)}\\ 
&({\rm IF})\ \ \ \frac{\sigma\models b}{(\mathbf{if}\ b\ \mathbf{then}\ C_1\ \mathbf{else}\ C_2\ \mathbf{fi},\sigma,|\varphi\rangle)\rightarrow (C_1,\sigma,|\varphi\rangle)}\quad\frac{\sigma\models\neg b}{(\mathbf{if}\ b\ \mathbf{then}\ C_1\ \mathbf{else}\ C_2\ \mathbf{fi},\sigma,|\varphi\rangle)\rightarrow (C_2,\sigma,|\varphi\rangle)}
\\
&({\rm QIF})\ \ \frac{|\psi\rangle=\alpha_0|0\rangle_{\sigma(q)}|\theta_0\rangle+\alpha_1|1\rangle_{\sigma(q)}|\theta_1\rangle\qquad (C_i,\sigma, |\theta_i\rangle)\rightarrow^\ast(\downarrow,\sigma^\prime, |\theta_i^\prime\rangle)\ (i=0,1)}{\left(\mathbf{qif}[q] \left(|0\rangle\rightarrow C_0\right)\square\left(|1\rangle\rightarrow C_1\right) \mathbf{fiq},\sigma, |\psi\rangle\right)\rightarrow \left(\downarrow,\sigma^\prime, \alpha_0|0\rangle_{\sigma(q)}|\theta_0^\prime\rangle+\alpha_1|1\rangle_{\sigma(q)}|\theta_1^\prime\rangle\right)}
\end{split}\end{equation*}
\caption{Transition Rules for Quantum Circuits with Classical Variables. In rule (AS), $\sigma[\overline{x}:=\sigma(\overline{t})]$ stands for the classical state obtained by updating the values of $\overline{x}$ with $\sigma(\overline{t})$ in $\sigma$.
In rule (GA), $U_{\sigma(q)}$ denotes that unitary operator acts on $\sigma(\overline{q})$, leaving other parts unchanged.}\label{circuit-rules-1}
\end{table*}
\end{defn}

The rules (SK), (SC), (AS) and (IF) are easy to understand. According to the rule (GA), the execution of unitary transformation $U[\overline{q}]$ in a circuit $C$ may fail if the premise $\mathit{Dist}(\overline{q})$ is not satisfied. More precisely, we say that unitary   
$U[\overline{q}]$ is an immediate command of $C$ if $C=U[\overline{q}];C^\prime$ or $C=\mathbf{qif}[q] \left(|0\rangle\rightarrow U[\overline{q}]; C_0\right)\square\left(|1\rangle\rightarrow C_1\right) \mathbf{fiq}$ or $C=\mathbf{qif}[q] \left(|0\rangle\rightarrow C_0\right)\square\left(|1\rangle\rightarrow U[\overline{q}]; C_1\right) \mathbf{fiq}$. A configuration $(C,\sigma,|\varphi\rangle)$ fails if $\sigma\not\models\mathit{Dist}(\overline{q})$ and some unitary $U[\overline{q}]$ is an immediate command of $C$. It eventually fails if $(C,\sigma,|\varphi\rangle)\rightarrow^\ast(C^\prime,\sigma^\prime,|\varphi^\prime\rangle)$ and $(C^\prime,\sigma^\prime,|\varphi^\prime\rangle)$ fails.  
It is easy to see that the execution failure of $C$ only depends on classical state $\sigma$ but not on quantum state $|\varphi\rangle$. So, we say that $(C,\sigma)$ is legitimate if for any $|\varphi\rangle$, $(C,\sigma,|\varphi\rangle)$ does not eventually fail. In addition, one design decision in the rule (QIF) needs an explanation. In its premise, when starting in the classical state $\sigma$, the executions of all branches $(C_i,\sigma,|\theta_i\rangle)$ $(i=0,1)$ are required to terminate in the same classical state $\sigma^\prime$. At the first glance, this is a very strong requirement and hard to meet in practical applications. Indeed, as we will see from rule (BS) in Table \ref{circuit-rules-1+} below, it can be easily achieved using local variables. 

\subsubsection{Local variables} 

In this subsection, we further introduce local classical variables into $\mathbf{QC}^+$ by extending its syntax with the following clause:
\begin{equation}\label{def-block}C::=\mathbf{begin\ local}\ \overline{x}:=\overline{t};C\ \mathbf{end},\end{equation} where $\overline{x}$ is a sequence of classical variables and $\overline{t}$ a sequence of classical expressions. A statement of the form (\ref{def-block}) is called a \textit{block statement}. It should be noted that the scope of local variables $\overline{x}$ in statement (\ref{def-block}) is explicitly given between $\mathbf{begin}$ and $\mathbf{end}$, and they are explicitly initialised by the assignment $\overline{x}:=\overline{t}$.  

The operational semantics of block statement (\ref{def-block}) is defined by the rule (BS) in Table \ref{circuit-rules-1+}. 
\begin{table*}[t]
\begin{equation*}({\rm BS})\ \ (\mathbf{begin\ local}\ \overline{x}:=\overline{t};C\ \mathbf{end}, \sigma, |\psi\rangle)\rightarrow (\overline{x}:=\overline{t};C;\overline{x}:=\sigma(\overline{x}),\sigma,|\psi\rangle)
\end{equation*}
\caption{Transition Rule for Local Variables.}\label{circuit-rules-1+}
\end{table*}
The rule (BS) is very similar to the rule defining the operational semantics of local variables in classical programs \cite{Apt09}. 
In the execution of block statement (\ref{def-block}) starting in classical state $\sigma$, the \textit{local variables} $\overline{x}$ are first initialised by the values $\sigma(\overline{t})$ of expressions $\overline{t}$, then the circuit $C$ within the statement is executed. After that, $\overline{x}$ resume their original values in $\sigma$. It is easy to see that if $\overline{x}$ is an empty sequence, then the block statement can be identified with circuit $C$ within it.  

\subsection{Quantum recursive programs without parameters}\label{sec-non-parameters}

Now we can define quantum recursive programs based on quantum circuits $\mathbf{QC}^+$ with classical variables. For a better understanding, let us first consider a simple class $\mathbf{RQC}^+$ of quantum recursive programs \textit{without} parameters.  

\subsubsection{Syntax} The alphabet of $\mathbf{RQC}^+$ is obtained by adding a set of \textit{procedure identifiers}, ranged over by symbols $P, P_1,P_2,...$ into the alphabet of $\mathbf{QC}^+$. Then the syntax of $\mathbf{RQC}^+$ is defined through  extending the syntax (\ref{def-syntax-1}) and (\ref{def-block}) of $\mathbf{QC}^+$ by adding the clause:
\begin{equation}
C::=P.\end{equation} It is assumed that quantum variables $\mathit{qv}(P)=\emptyset$. As in classical recursive programming (see for example, Chapter 4 of \cite{Apt09}), an occurrence of a procedure identifier in a program is called a \textit{procedure call}. We assume that each procedure identifier $P$ is defined by a \textit{declaration} of the form \begin{equation}\label{def-sem-2}P\Leftarrow C,\end{equation} where $C\in\mathbf{RQC}^+$ is called the procedure body. Note that in the declaration (\ref{def-sem-2}), $P$ may appear in the procedure body $C$; the occurrences of $P$ in $C$ are thus called \textit{recursive calls}. We assume a fixed set $\mathcal{D}$ of procedure declarations.  

We remember that in a quantum if-statement, a quantum coin should be external to the quantum circuits controlled by it. To enforce this rule in the circumstance of recursion, we introduce: 
\begin{defn}\label{def-well}A quantum recursive circuit $C\in\mathbf{RQC}^{+}$ is said to be well-defined if for any sub-circuit $\mathbf{qif}[q]\ (|0\rangle\rightarrow C_0)\square(|1\rangle\rightarrow C_1)\ \mathbf{fiq}$ of $C$, and for any procedure identifier $P$ in $C_0$ or $C_1$ with declaration $P\Leftarrow D\in\mathcal{D}$, we have $q\notin\mathit{qv}(D)$.\end{defn} 

The requirement in the above definition comes for compatibility with the condition $q\notin\mathit{qv}(C_0)\cup\mathit{qv}(C_1)$ in Definition \ref{def-circ}(3). In the sequel, all quantum circuits are assumed to be well-defined. 

\subsubsection{Operational semantics}To define the operational semantics of $\mathbf{RQC}^+$, we first generalise the notion of configuration $(C,\sigma, |\psi\rangle)$ by allowing $C\in\mathbf{RQC}^+$. Then the operational semantics is the transition relation $\rightarrow$ between configurations defined by the transition rules given in Tables \ref{circuit-rules-1} and \ref{circuit-rules-1+} together with the \textit{copy rule} (CR) in Table \ref{circuit-rules-2-0}. 
\begin{table*}[t]
\begin{equation}({\rm CR})\qquad \frac{P\Leftarrow C\in\mathcal{D}}{(P,\sigma, |\psi\rangle)\rightarrow (C,\sigma, |\psi\rangle)}\end{equation}
\caption{Transition Rule for Quantum Recursive Circuits without Parameters.}\label{circuit-rules-2-0}
\end{table*} 
Intuitively, the copy rule (CR) allows that a procedure call is dynamically replaced by the procedure body of its declaration given in declarations $\mathcal{D}$; that is, the procedure body is copied at the point of a call during the execution of a recursive program.  

\subsection{Quantum recursive programs with parameters}\label{sec-parameters}

In this subsection, we further expand the language $\mathbf{RQC}^+$ defined in the last subsection to $\mathbf{RQC}^{++}$ of quantum recursive circuits with \textit{classical} parameters.

\subsubsection{Syntax} We add a set of procedure identifiers $P, P_1,P_2,...$ into the alphabet of $\mathbf{QC}^{+}$. Each identifier $P$ is given an arity $\mathit{ar}(P)$. 
Then the syntax of $\mathbf{RQC}^{++}$ is defined by the syntax (\ref{def-syntax-1}) and (\ref{def-block}) of $\mathbf{QC}^+$ together with the following clause: 
\begin{align}\label{def-syntax-2}C::=\ & P(\overline{t})\end{align}
The above syntax of $\mathbf{RQC}^{++}$ is very similar to that of classical recursive programs with parameters given in \cite{Apt09}. In procedure call (\ref{def-syntax-2}), $P$ is a procedure identifier with $\mathit{ar}(P)=n$ and $\overline{t}=t_1,...,t_n$, where $t_1,...,t_n$ are classical expressions, called \textit{actual parameters}. Whenever $n=0$, procedure $P(\overline{t})$ degenerates to a procedure without parameters considered in the last subsection. 

As usual, each procedure identifier $P$ of $\mathit{ar}(P)=n$ is defined by a declaration of the form: \begin{equation}P(\overline{u})\Leftarrow C,\end{equation} where $\overline{u}=u_1,...,u_n$ are classical simple variables, called \textit{formal parameters}; and the procedure body $C\in\mathbf{RQC}^{++}.$ We assume a fixed set $\mathcal{D}$ of procedure declarations. The notion of well-defined quantum recursive circuit in Definition \ref{def-well} can be straightforwardly generalized to $\mathbf{QC}^{++}$. 

To demonstrate its expressive power, let us give a simple example that shows how controlled gate by multiple qubits can be properly programmed in our language $\mathbf{RQC}^{++}$: 

\begin{exam}[Controlled Gates]\label{exam-control-1} Let $U$ be a unitary operator. Then the controlled-$U$ gate on the section $q[m:n]$ with the first $(n-m)$ qubits as its control qubits and the last one as its target qubit is defined by \begin{equation}\label{def-controlled}C^{(\ast)}(U)|i_m,...,i_{n-1}\rangle|\psi\rangle =\begin{cases}
|i_m,...,i_{n-1}\rangle U|\psi\rangle\ &{\rm if}\ i_m=...=i_{n-1}=1;\\ |i_m,...,i_{n-1}\rangle|\psi\rangle &{\rm otherwise}
\end{cases}\end{equation} for any $i_m,...,i_{n-1}\in\{0,1\}$ and $|\psi\rangle\in\hs_2$.

The controlled gate (\ref{def-controlled}) can be written as the following recursive program:
\begin{equation}\label{c-gate}\begin{split}C^{(\ast)}(U)(m,n)\ \Leftarrow\ &\mathbf{if}\ m=n\\ &\ \ \ \mathbf{then}\ U[q[n]]\\
&\ \ \ \mathbf{else}\ \mathbf{qif}[q[m]]|0\rangle\rightarrow\mathbf{skip}\\ 
&\quad\ \qquad\qquad\square\ \ \ |1\rangle\rightarrow C^{(\ast)}(U)(m+1,n)\\ &\qquad\quad\ \mathbf{fiq}\\ 
&\mathbf{fi}.
\end{split}\end{equation}    
\end{exam}

A more sophisticated example is our motivating example of quantum Fourier transform (Example \ref{example-QFT}). Some other examples will be presented in Appendix \ref{appendix-examples}.

\subsubsection{Operational semantics} Similar to Definition \ref{def-well}, we can define the semantics of quantum recursive circuits $\mathbf{RQC}^{++}$ with parameters. 
Now configurations are triples $(C,\sigma,|\psi\rangle)$ with $C\in\mathbf{RQC}^{++}$ and $\sigma, |\psi\rangle$ are the same as in Section \ref{sec-non-parameters}. The semantics of $\mathbf{RQC}^{++}$ is then a transition relation between configurations defined by the rules in Tables \ref{circuit-rules-1} and \ref{circuit-rules-1+} together with the recursive rule (RC) in Table \ref{circuit-rules-2}. 
\begin{table*}[t]
\begin{equation*}({\rm RC})\ \ \frac{P(\overline{u})\Leftarrow C\in\mathcal{D}}{(P(\overline{t}),\sigma, |\psi\rangle)\rightarrow (\mathbf{begin\ local}\ \overline{u}:=\overline{t};C\ \mathbf{end},\sigma, |\psi\rangle)}
 \end{equation*}
\caption{Transition Rule for Quantum Recursive Circuits with Parameters.}\label{circuit-rules-2}
\end{table*} 
The transition rule (RC) for quantum recursive circuits with parameters is essentially a copy rule. The main difference between it and the copy rule (CR) in Table \ref{circuit-rules-2-0} is that a block statement around the procedure body $C$ is introduced to limit the scope of formal parameters $\overline{u}$. This block statement guarantees that the values of the formal parameters $\overline{u}$ are not changed at the end of the procedure call, and the formal parameters $\overline{u}$ are no longer accessible after the procedure call terminates.       

 The following lemmas shows that for any given initial state $\sigma$ of classical variables, then (the operational semantics of) each quantum circuit $C\in\mathbf{RQC}^{++}$ actually induces a unitary operator $\llbracket C\rrbracket_\sigma$ on the quantum variables, provided $(C,\sigma)$ is legitimate (see the paragraph after Definition \ref{def-sem-classical}). Moreover, the family $\llbracket C\rrbracket=\left\{\llbracket C\rrbracket_\sigma\right\}$ of unitary operators indexed by the states of classical variables can be understood as the denotational semantics of $C$.    
\begin{lem}\label{lem-linear1}
For any quantum circuit $C\in\mathbf{RQC}^{++}$ and for any classical state $\sigma$, we have: 
\begin{enumerate}\item If for some $|\varphi_0\rangle\in\hs_{\mathit{qv}(C)}$, it holds that $(C,\sigma,|\varphi_0\rangle)\rightarrow^\ast (\downarrow,\sigma^\prime,|\psi_0\rangle)$ for some $\sigma^\prime$ and $|\psi_0\rangle$, then for all $|\varphi\rangle\in\hs_{\mathit{qv}(C)}$, we have for some $|\psi\rangle$: \begin{equation}\label{det-computing}(C,\sigma,|\varphi\rangle)\rightarrow^\ast (\downarrow,\sigma^\prime,|\psi\rangle).\end{equation} \item In this case, we define mapping $\llbracket C\rrbracket_\sigma:\hs_{\mathit{qv}(C)}\rightarrow\hs_{\mathit{qv}(C)}$ by 
\begin{equation}\label{det-computing1}\llbracket C\rrbracket_\sigma(|\varphi\rangle)=|\psi\rangle\ {\rm if}\ (C,\sigma,|\varphi\rangle)\rightarrow^\ast(\downarrow,\sigma^\prime,|\psi\rangle).\end{equation} Then $\llbracket C\rrbracket_\sigma$ is well-defined and it is a unitary operator on $\hs_{\mathit{qv}(C)}.$
\end{enumerate}\end{lem}
\begin{proof} The proof is routine by induction on the length of $C$. We omit the tedious details here. 
\end{proof}

It should be noted from the operational semantics of $\mathbf{RQC}^{++}$ that the computation (\ref{det-computing}) is deterministic, since the underlying classical programming language is assumed to be deterministic. Therefore, the terminal classical state $\sigma^\prime$ in equation (\ref{det-computing1}) is uniquely determined by the initial classical state $\sigma$. This is why the notation $\llbracket C\rrbracket_\sigma$ has only the subscript $\sigma$ but not $\sigma^\prime$.   

\section{Correctness Formulas}\label{sec-correctness}

In this section, we define Hoare triples for specifying correctness of quantum recursive circuits in $\mathbf{RQC}^{++}$ described in the previous section. 

\subsection{Parameterized Quantum States}\label{sec-parameter}

Since a quantum circuit $C\in\mathbf{RQC}^{++}$ is defined with classical variables, as already mentioned in Section \ref{Contri}, quantum states $|\varphi\rangle$ in the precondition and $|\psi\rangle$ in the postcondition of correctness formula (\ref{Hoare-00}) must be parameterized by the states of classical variables. Formally, let $q:T_1\times ...\times T_n\rightarrow\hs$ be a array quantum variable. Then for any values $v_1\in T_1,...,v_n\in T_n$, array element $q^\prime=q[v_1,...,v_n]$ has the Hilbert space $\hs_{q^\prime}=\hs$. We write $E(q)=\{q[v_1,...,v_n]|v_1\in T_1,...,v_n\in T_n\}$. We further write $E(\mathit{QV})=\bigcup_{q\in\mathit{QV}}E(q)$ for the set of all quantum simple variables and array elements in $\mathit{QV}$. Then we have:

\begin{defn}Let $\Sigma$ be the set of states of classical variables, and let $S:\Sigma\rightarrow 2^{E(\mathit{QV})}$ be a mapping from classical states to subsets of quantum variables. Then a parameterized quantum state of signature $S$ is a (partial) function $|\varphi\rangle:\sigma\mapsto |\varphi\rangle(\sigma)$ such that for each $\sigma\in\Sigma$, if $|\varphi\rangle(\sigma)$ is defined then it a pure state in a Hilbert space $\hs_{S(\sigma)}=\bigotimes_{q\in S(\sigma)}\hs_q$ of quantum variables in $S(\sigma)$.\end{defn}

We use the same notation $|\varphi\rangle, |\psi\rangle,...$ for ordinary and parameterized quantum states since every quantum state $|\varphi\rangle$ can be seen as a special parameterized quantum state, also denoted $|\varphi\rangle$, defined by $|\varphi\rangle(\sigma)\equiv |\varphi\rangle$ for all classical states $\sigma$. Various operations of quantum states can be extended point-wise to parameterized quantum states. For example, if $|\varphi\rangle,|\psi\rangle$ are parameterized quantum states, $\alpha,\beta$ are complex numbers, and $U$ is a unitary operator, then $\alpha |\varphi\rangle+\beta|\psi\rangle$ is defined by $(\alpha |\varphi\rangle+\beta|\psi\rangle)(\sigma)=\alpha |\varphi\rangle(\sigma)+\beta|\psi\rangle(\sigma)$, and $U|\varphi\rangle$ is defined by $(U|\varphi\rangle)(\sigma)=U|\varphi\rangle(\sigma)$ for all $\sigma$. 

Let $|\varphi\rangle$ and $|\psi\rangle$ are two parameterized quantum states. For a given classical state $\sigma\in\Sigma$, we say that $|\varphi\rangle$ equals $|\psi\rangle$ in $\sigma$, written $\sigma\models |\varphi\rangle=|\psi\rangle$, if $|\varphi\rangle(\sigma)=|\psi\rangle(\sigma)$. 
 
 \begin{exam}\label{exam-GHZ}\begin{enumerate}\item Consider qubit array section $q[m:n]$. Then the equal superposition and the GHZ state over it can be seen as parameterized quantum states:
  \begin{align*}&|S\rangle: (m,n)\mapsto |S\rangle(m,n)=\frac{1}{\sqrt{2^{n-m+1}}}\sum_{x\in\{0,1\}^{n-m+1}}|x\rangle\in\bigotimes\hs_2^{\otimes (n-m+1)},\\  &|\mathit{GHZ}\rangle:(m,n)\mapsto |\mathit{GHZ}\rangle(m,n)=\frac{1}{\sqrt{2}}\left(|0\rangle^{\otimes (n-m+1)}+|1\rangle^{\otimes (n-m+1)}\right)\in\bigotimes\hs_2^{\otimes (n-m+1)}.\end{align*}
\item If a classical state $\sigma$ satisfies $\sigma(m)=\sigma(n)$, then we have $\sigma\models |S\rangle=|\mathit{GHZ}\rangle$.
\end{enumerate}\end{exam}

We can further define a syntax for parameterized quantum states so that they can be described in a formal language (and thus symbolized). Due to the limitation of space, we defer it to Appendix \ref{formal-states}. Here, we choose to use the above definition because it is easy to use and understand for our application examples given in Section \ref{sec-examples} and Appendix \ref{appendix-examples}. 

\subsection{Quantum Hoare Triples}

Now we can define the notion of correctness formula, usually called Hoare triple, for specifying the correctness of quantum circuits with classical variables as their parameters.  

\begin{defn}A Hoare triple for quantum recursive circuits is defined of the form \begin{equation}\label{Hoare-0+}\{A,|\varphi\rangle\}\ C\ \{B,|\psi\rangle\},\end{equation} where: \begin{enumerate}\item $C\in\mathbf{RQC}^{++}$ is a quantum recursive circuit; \item $A$ and $B$ are first-order logical formulas over classical variables; and \item $|\varphi\rangle$ and $|\psi\rangle$ are parameterized quantum states. 
\end{enumerate}
\end{defn}

The precondition $\{A,|\varphi\rangle\}$ and postcondition $\{B,|\psi\rangle\}$ are both \textit{classical-quantum assertions} in the following sense: The logical formulas $A$ in the precondition $\{A,|\varphi\rangle\}$ and $B$ in the postcondition $\{B,|\psi\rangle\}$ are used to specify the properties of the initial and terminal states, respectively, of classical variables, whereas $|\varphi\rangle$ and $|\psi\rangle$ are used to denote the initial and terminal states, respectively, of quantum variables (parameterized by classical variables).  

Intuitively, Hoare triple (\ref{Hoare-0+}) means that if we start to execute quantum circuit $C$ from an initial classical state $\sigma$ that satisfies precondition $A$, then $C$ transforms quantum state $|\varphi\rangle(\sigma)$ to $|\psi^\prime\rangle$, and at the end of the execution, the classical variables must be in a state $\sigma^\prime$ that satisfies postcondition $B$ and $|\psi^\rangle=|\psi\rangle(\sigma^\prime)$. As usual, we distinguish partial and total correctness based on whether the termination of execution is required or not. Formally, we first note that the classical programming language assumed in this paper should be defined upon a first-order logical language, called the \textit{assertion language}. Accordingly, the semantics of Hoare triple (\ref{Hoare-0+}) is defined upon the semantics of this assertion language. Then we have:  

\begin{defn}\label{def-hoare-semantics} Given an interpretation $\mathbb{I}$ of the assertion language. \begin{enumerate}\item A Hoare triple $\{A,|\varphi\rangle\}\ C\ \{B,|\psi\rangle\}$ is true in $\mathbb{I}$ in the sense of partial correctness, written $$\mathbb{I} \models_\mathit{par}\{A,|\varphi\rangle\}\ C\ \{B,|\psi\rangle\},$$ if for any state $\sigma$ of classical variables, whenever $\sigma\models_\mathbb{I} A$ (in the sense of first-order logic), and $(C,\sigma,|\varphi\rangle(\sigma))\rightarrow^\ast (\downarrow,\sigma^\prime,|\psi^\prime\rangle)$ for some state $\sigma^\prime$ of classical variables and for some quantum state $|\psi^\prime\rangle$, then we have:\begin{enumerate} 
\item $\sigma^\prime\models_\mathbb{I} B;$ and \item $|\psi^\prime\rangle=|\psi\rangle(\sigma^\prime)$. 
\end{enumerate}
\item A Hoare triple $\{A,|\varphi\rangle\}\ C\ \{B,|\psi\rangle\}$ is true in the sense of total correctness, written $$\mathbb{I}\models_\mathit{tot}\{A,|\varphi\rangle\}\ C\ \{B,|\psi\rangle\},$$ if for any state $\sigma$ of classical variables, whenever $\sigma\models_\mathbb{I} A$, then for some state $\sigma^\prime$ of classical variables and for some quantum state $|\psi^\prime\rangle$, we have:\begin{enumerate}\item $(C,\sigma,|\varphi\rangle(\sigma))\rightarrow^\ast (\downarrow,\sigma^\prime,|\psi\rangle)$; 
\item $\sigma^\prime\models_\mathbb{I} B;$ and \item $|\psi^\prime\rangle=|\psi\rangle(\sigma^\prime)$. 
\end{enumerate}
\end{enumerate}
\end{defn}

We need several more notations in the subsequent discussions. By the judgement: $$\models_\mathit{par}\{A,|\varphi\rangle\}\ C\ \{B,|\psi\rangle\},$$ we mean that $\mathbb{I} \models_\mathit{par}\{A,|\varphi\rangle\}\ C\ \{B,|\psi\rangle\}$ for all interpretations $\mathbb{I}$ of the assertion language. Similarly, we can define $\models_\mathit{tot}\{A,|\varphi\rangle\}\ C\ \{B,|\psi\rangle\}$. Furthermore, we often write: $$A\models_\mathit{par}\{|\varphi\rangle\}\ C\ \{|\psi\rangle\}\ ({\rm respectively,}\ A\models_\mathit{tot}\{|\varphi\rangle\}\ C\ \{|\psi\rangle\})$$ if it holds that $\models_\mathit{par}\{A,|\varphi\rangle\}\ C\ \{\mathit{true},|\psi\rangle\}$ (respectively, $\models_\mathit{tot}\{A,|\varphi\rangle\}\ C\ \{\mathit{true},|\psi\rangle\}$). 

To illustrate the notions of correctness introduced above, let us consider a simple example: 

\begin{exam}\label{exam-control-2} Let us introduce a classical bit array variable $i$ of the type $\mathbf{integer}\rightarrow\mathbf{Boolean}$. We define two parameterized quantum states:
\begin{align*}&|\Phi\rangle:(i,m,n)\mapsto|\Phi\rangle(i,m,n)=|i[m:n]\rangle\in\hs_2^{\otimes (n-m+1)}\\
&|\Psi\rangle:(i,m,n)\mapsto|\Psi\rangle(i,m,n)=|i[m:n-1]\rangle U|i[n]\rangle\in\hs_2^{\otimes (n-m+1)},\end{align*}
where we use the convention that the $|i[m:n-1]\rangle$ part is omitted if $m>n-1$.
Then the correctness of the controlled gate in Example \ref{exam-control-1}, i.e. quantum recursive program (\ref{c-gate}), can be specified by the following two Hoare triples: 
\begin{align}\label{c-gate-1}&(m\leq n)\wedge(\forall m\leq k<n)(i[k]=1)\models_\mathit{tot} \left\{|\Phi\rangle\right\}\ C^{(\ast)}(m,n)\ \left\{|\Psi\rangle\right\};\\
\label{c-gate-2}&(m\leq n)\wedge(\exists m\leq k<n)(i[k]=0)\models_\mathit{tot} \left\{|\Phi\rangle\right\}\ C^{(\ast)}(m,n)\ \left\{|\Phi\rangle\right\}.\end{align}\end{exam}

A proof of (\ref{c-gate-1}) and (\ref{c-gate-2}) using our proof system given in the next section is presented in Appendix \ref{appendix-examples}.

\section{Proof System for Quantum Recursive Circuits}\label{proof-system}

In this section, we present a set of inference rules for proving the correctness of quantum circuits recursively defined in the programming language introduced in Section \ref{sec-language}. These rules will be organised into two proof (sub)systems, one for partial correctness and one for total correctness.   

\subsection{Structural Rules}\label{sec-struct-proof}

The first group of our proof rules are three structural rules shown in Table \ref{circuit-rules-proof-0}: (Frame), (Linearity) and (Consequence). 
\begin{table*}[t]
\begin{equation*}\begin{split}&({\rm Frame})\ \ \ \frac{\{A,|\varphi\rangle\}\ C\ \{B,|\psi\rangle\}\qquad \overline{q}\cap\mathit{qv}(C)=\emptyset}{\left\{A,|\varphi\rangle|\theta\rangle_{\overline{q}}\right\}\ C\ \left\{B,|\psi\rangle|\theta\rangle_{\overline{q}}\right\}}\\ 
&({\rm Linearity})\ \ \frac{\left\{A,|\varphi_i\rangle\right\}\ C\ \left\{B,|\psi_i\rangle\right\}\ (i=1,2)}{\left\{A,\alpha_1|\varphi_1\rangle+\alpha_2|\varphi_2\rangle\right\}\ C\ \left\{B,\alpha_1|\psi_1\rangle+\alpha_2|\psi_2\rangle\right\}}
 \\ &({\rm Consequence})\ \ \frac{\{A,|\varphi\rangle\}\models \{A^\prime,|\varphi^\prime\rangle\}\qquad \{A^\prime,|\varphi^\prime\rangle\}\ C\ \{B^\prime,|\psi^\prime\rangle\}\qquad \{B^\prime,|\psi^\prime\rangle\}\models \{B,|\psi\rangle\}}{\{A,|\varphi\rangle\}\ C\ \{B,|\psi\rangle\}}
\end{split}\end{equation*}
\caption{Structural Rules. 
}\label{circuit-rules-proof-0}
\end{table*}
Intuitively, they are explained as follows: \begin{itemize}\item The rule (Frame) allows us to adapt a proved correctness into a larger environment with more quantum variables. It is interesting to note that a similar adaptation can be freely done in classical Hoare logic without an explicit statement of frame rule.  
\item The rule (Linearity) is introduced to define the linearity of quantum circuits. 
\item As in the classical Hoare logic, the rule (Consequence) enables us to strengthen the precondition and to weaken the postcondition in a Hoare triple. But the entailment $\models$ used here is defined between two classical-quantum assertions:  
\end{itemize}

\begin{defn}We say that condition $\{A,|\varphi\rangle\}$ entails condition $\{B,|\psi\rangle\}$, written $$\{A,|\varphi\rangle\}\models\{B,|\psi\rangle\},$$ if $A\models B$ and $A\models |\varphi\rangle=|\psi\rangle$ (that is, for any $\sigma\models A$, we have $\sigma\models |\varphi\rangle=|\psi\rangle$, i.e. $|\varphi\rangle(\sigma)=|\psi\rangle(\sigma)$).
\end{defn}

Here, $A\models |\varphi\rangle=|\psi\rangle$ should be understood as follows: Both $|\varphi\rangle$ and $|\psi\rangle$ may be given with certain classical variables as their parameters. Then $A\models |\varphi\rangle=|\psi\rangle$ means that whenever the classical variables in $|\varphi\rangle$ and $|\psi\rangle$ satisfy condition $A$, then $|\varphi\rangle$ and $|\psi\rangle$ are the same. 

\subsection{Proof Rules for Basic Program Constructs}

The second group of our proof rules are six construct-specific rule given in Table \ref{circuit-rules-proof-1}, each of which is designed for proving the correctness of one program construct defined in the syntax (\ref{def-syntax-1}). 
\begin{table*}[t]
\begin{equation*}\begin{split}&({\rm SK-P})\ \ \{A,|\varphi\rangle\}\ \mathbf{skip}\ \{A,|\varphi\rangle\}\qquad\qquad ({\rm AS-P})\ \ \left\{A[\overline{x}:=\overline{t}],|\varphi\rangle\right\}\ \overline{x}:=\overline{t}\ \{A,|\varphi\rangle\}\\ 
&({\rm GA-P})\ \ \left\{A, U^\dag|\varphi\rangle\right\}\ U[\overline{q}]\ \{A,|\varphi\rangle\}\\ 
&({\rm SC-P})\ \ \frac{\{A,|\varphi\rangle\}\ C_1\ \{B,|\psi\rangle\}\qquad \{B,|\psi\rangle\}\ C_2\ \{D,|\theta\rangle\}}{\{A,|\varphi\rangle\}\ C_1;C_2\ \{D,|\theta\rangle\}}\\
&({\rm IF-P})\ \ \frac{\{A\wedge b,|\varphi\rangle\}\ C_1\ \{B,|\psi\rangle\}\qquad \{A\wedge\neg b,|\varphi\rangle\}\ C_2\ \{B,|\psi\rangle\}}{\{A,|\varphi\rangle\}\ \mathbf{if}\ b\ \mathbf{then}\ C_1\ \mathbf{else}\ C_2\ \mathbf{fi}\ \{B,|\psi\rangle\}}\\
&({\rm QIF-P})\ \ \frac{\{A,|\varphi_i\rangle\}\ C_i\ \{B,|\psi_i\rangle\}\qquad \mathit{cv}(q)\cap\mathit{change}(C_i)=\emptyset\ (i=0,1)}{\begin{array}{ccc}\left\{A,\alpha_0|0\rangle_q|\varphi_0\rangle+\alpha_1|1\rangle_q|\varphi_1\rangle\right\}\qquad\qquad\qquad\qquad\qquad\qquad\qquad\qquad\qquad\\ \mathbf{qif}[q] (|0\rangle\rightarrow C_0)\ \square\ (|1\rangle\rightarrow C_1)\ \mathbf{fiq}\\ \qquad\qquad\qquad\qquad\qquad\qquad\qquad\qquad\qquad\left\{B,\alpha_0|0\rangle_q|\psi_0\rangle+\alpha_1|1\rangle_q|\psi_1\rangle\right\}\end{array}}\\
\end{split}\end{equation*}
\caption{Proof Rules for Quantum Circuits with Classical Variables.}\label{circuit-rules-proof-1}
\end{table*}

The first five rules in Table \ref{circuit-rules-proof-1} are easy to understand. Indeed, they are the corresponding rules in either classical or quantum Hoare logic adapted to the Hoare triples of the form (\ref{Hoare-00}) for quantum circuits with classical variables as their parameters. 

The rule (QIF-P) in Table \ref{circuit-rules-proof-1} has not been introduced in the previous literature. It characterises the behaviour of quantum if-statement. Here, $\mathit{cv}(q)$ denotes the set of classical variables appearing in $q$; that is, if $q$ is a simple quantum variable, then $\mathit{cv}(q)=\emptyset$, and if $q=q^\prime[t_1,...,t_n]$ is a subscripted quantum variable, then $\mathit{cv}(q)=\bigcup_{i=1}^n\mathit{var}(t_i)$. Also, $\mathit{change}(C_i)$ denotes the set of classical variables changed by program $C_i$. Indeed, the rule (QIF-P) can be simplified to   
\begin{equation*}
   ({\rm QIF-P-S})\ \ \frac{\{A,|\varphi\rangle\}\ C_i\ \{B,|\psi\rangle\}\qquad \mathit{cv}(q)\cap\mathit{change}(C_i)=\emptyset\ (i=0,1)}{\{A,|i\rangle|\varphi\rangle\}\ \mathbf{qif}[q] (|0\rangle\rightarrow C_0)\ \square\ (|1\rangle\rightarrow C_1)\ \mathbf{fiq}\ \{B,|i\rangle|\psi\rangle\}}\ (i=0,1)
\end{equation*} Obviously, (QIF-P) can be derived from (QIF-P-S) with the rule (Linearity). But (QIF-P) is more convenient in practical applications.

\subsection{Proof Rules for Parameterisation}\label{sec-rule-para}

The quantum recursive programs considered in this paper are defined using classical terms as parameters. So, we introduce two proof rules in Table \ref{circuit-rules-proof-2} to handle this mechanism of parameterisation.
\begin{table*}[t]
\begin{equation*}\begin{split}&({\rm Block})\ \ \frac{\{A,|\varphi\rangle\}\ \overline{x}:=\overline{t};C\ \{B,|\psi\rangle\}\qquad \overline{x}\cap\mathit{free}(B) =\emptyset\qquad \mathit{Inv}(|\psi\rangle: \overline{x})}{\{A,|\varphi\rangle\}\ \mathbf{begin\ local}\ \overline{x}:=\overline{t};C\ \mathbf{end}\ \{B,|\psi\rangle\}}\\ 
\\ 
&({\rm Instantiation})\ \ \frac{ 
\{A,|\varphi\rangle\}\ P(\overline{x})\ \{B,|\psi\rangle\}\qquad \overline{x}\cap\mathit{var}(\mathcal{D})=\mathit{var}(\overline{t})\cap\mathit{change}(\mathcal{D})=\emptyset}{\{A[\overline{x}:=\overline{t}],|\varphi\rangle[\overline{x}:=\overline{t}]\}\ P(\overline{t})\ \{B[\overline{x}:=\overline{t}],|\psi\rangle[\overline{x}:=\overline{t}]\}}\\ \\ 
\end{split}\end{equation*}
\caption{Proof Rules for Blocks and Instantiation.}\label{circuit-rules-proof-2}
\end{table*} Their intuitive meanings are explained as follows:
\begin{itemize}\item The rule (Block) in Table \ref{circuit-rules-proof-2} enables us to localise a family of classical variables in a block statement provided that they are not free variables in the postcondition of an intended correctness formula. Here, (i) $\mathit{free}(B)$ stands for the set of free (classical) variables in first-order logical formula $B$; and (ii)  $\mathit{Inv}(|\psi\rangle: \overline{x})$ means
that $|\psi\rangle$ is invariant over $\overline{x}$; that is, $|\psi\rangle(\sigma_1)=|\psi\rangle(\sigma_2)$ for all classical states $\sigma_1$ and $\sigma_2$ that coincide everywhere except $\overline{x}$.   
\item The rules (Instantiation) in Table \ref{circuit-rules-proof-2} allows us to instantiating free (classical) variables with terms under certain conditions for the freedom of variable conflicts. Here, (a) $A[\overline{x}:=\overline{t}]$ and $B[\overline{x}:=\overline{t}]$ stand for the substitution of classical variables $\overline{x}$ by classical expressions $\overline{t}$ in first-order logical formulas $A$ and $B$, respectively; (b) For a classical state $\sigma$, a sequence $\overline{x}$ of classical variables and a sequence of classical values, $\sigma[\overline{x}:=\overline{v}]$ denotes the classical state defined by $\sigma[\overline{x}:=\overline{v}](y)=\sigma(y)$ if $y\notin\overline{x}$ and $\sigma[\overline{x}:=\overline{v}](x)=v$ if $x\in\overline{x}$. 
Then $|\varphi\rangle[\overline{x}:=\overline{t}]$ denotes the parameterized quantum state defined by $(|\varphi\rangle[\overline{x}:=\overline{t}])(\sigma)=|\varphi\rangle(\sigma[\overline{x}:=\sigma(\overline{t})])$ for every classical state $\sigma$. 
The parameterized quantum state $|\psi\rangle[\overline{x}:=\overline{t}]$ is defined similarly;  
and (c) $\mathit{var}(\mathcal{D})$ and $\mathit{change}(\mathcal{D})$ stand for the sets of classical variables that appear in and are  changed, respectively, by the programs in $\mathcal{D}$.
\end{itemize}

\subsection{Proof Rules for Quantum Recursion}\label{sec-rec-proof}

In this subsection, we introduce two key rules in our proof system for verification of the correctness of quantum recursive circuits. They are given in Table \ref{circuit-rules-proof-3} for reasoning about recursion in quantum circuits: (Recursion-Par) for partial correctness and (Recursion-Tot) for total correctness. 
These two rules need some careful explanations: \begin{itemize}
\item Note that provability notation $\vdash$ appears in the premises of (Recursion-Par) and (Recursion-Tot). For any set $\Sigma$ of Hoare triples, by the judgement: \begin{equation}\label{provability}\Sigma\vdash\{A,|\varphi\rangle\}\ C\ \{B,|\psi\rangle\}\end{equation} we mean that $\{A,|\varphi\rangle\}\ C\ \{B,|\psi\rangle\}$ can be derived from assumptions $\Sigma$ using the proof rules given in Tables \ref{circuit-rules-proof-0} and \ref{circuit-rules-proof-1}. In particular, $\vdash$ does not depend on the declarations $\mathcal{D}$.   

\item It is easy to see that the essential part of (Recursion-Par) is actually the special case:  
\begin{equation*}({\rm Recursion-Par-E})\ \ \frac{\begin{split}&\{A_i,|\varphi_i\rangle\}\ P_i(\overline{t_i})\ \{B_i,|\psi_i\rangle\} (i=1,...,n)\vdash \{A_j,|\varphi_j\rangle\}\\ &\qquad\qquad\qquad\qquad\mathbf{begin\ local}\ \overline{u_j}:=\overline{t_j};C_j\ \mathbf{end}\ \{B_j,|\psi_j\rangle\}\ {\rm for}\ j=1,...,n
\end{split}}{\{A_j, |\varphi_j\rangle\}\ P_j(\overline{t_j})\ \{B_j,|\psi_j\rangle\}\ {\rm for}\ j=1,...,n}\end{equation*}
The rule (Recursion-Par-E) means that if one can establish the correctness of each procedure body $C_j$ (in a block) from the assumption of the correctness of all procedure calls $P_i(\overline{t_i})$ $(i=1,...,n)$, then the correctness of each procedure call $P_j(\overline{t_j})$ is proved.    

\item The only difference between the rules (Recursion-Tot) for total correctness and (Recursion-Par) for partial correctness is the \textit{ranking function} $r$, which is used in the standard way as in reasoning about total correctness of classical programs. 
\end{itemize}

\begin{table*}[t]
\begin{equation*}\begin{split}&({\rm Recursion-Par})\ \ \frac{\begin{split}&\{A_i,|\varphi_i\rangle\}\ P_i(\overline{t_i})\ \{B_i,|\psi_i\rangle\} (i=1,...,n)\vdash \{A,|\varphi\rangle\}\ C\ \{B,|\psi\rangle\}\\ &\{A_i,|\varphi_i\rangle\}\ P_i(\overline{t_i})\ \{B_i,|\psi_i\rangle\} (i=1,...,n)\vdash \{A_j,|\varphi_j\rangle\}\\ &\qquad\qquad\qquad\qquad\mathbf{begin\ local}\ \overline{u_j}:=\overline{t_j};C_j\ \mathbf{end}\ \{B_j,|\psi_j\rangle\}\ {\rm for}\ j=1,...,n
\end{split}}{\{A, |\varphi\rangle\}\ C\ \{B,|\psi\rangle\}}\\ \\ 
&({\rm Recursion-Tot})\ \ \frac{\begin{split}&\{A_i,|\varphi_i\rangle\}\ P_i(\overline{t_i})\ \{B_i,|\psi_i\rangle\} (i=1,...,n)\vdash \{A,|\varphi\rangle\}\ C\ \{B,|\psi\rangle\}\\ &\{A_i\wedge r<z,|\varphi_i\rangle\}\ P_i(\overline{t_i})\ \{B_i,|\psi_i\rangle\} (i=1,...,n)\vdash \{A_j\wedge r=z,|\varphi_j\rangle\}\\ &\qquad\qquad\qquad\qquad\qquad\mathbf{begin\ local}\ \overline{u_j}:=\overline{t_j};C_j\ \mathbf{end}\ \{B_j,|\psi_j\rangle\}\ {\rm for}\ j=1,...,n\\ 
& A_i\rightarrow r\geq 0\ {\rm for}\ i=1,...,n\end{split}}{\{A, |\varphi\rangle\}\ C\ \{B,|\psi\rangle\}}\end{split}\end{equation*}
\caption{Proof Rules for Quantum Recursive Circuits. In rules (Recursion-Par) and (Recursion-Tot), it is assumed that $\mathcal{D}=\left\{P(\overline{u_i})\Leftarrow C_i\mid i=1,...,n)\right\}$. In rule (Recursion-Tot), $r$ is an integer expression, called \textit{ranking function}, and $z\notin\mathit{var}(r)\cup\bigcup_{i=1}^n[\mathit{free}(A_i)\cup\mathit{free}(B_i)\cup \mathit{var}(C_i)]$ is an integer variable.}\label{circuit-rules-proof-3}
\end{table*}

\subsection{Auxiliary Rules}

In classical programming, it was pointed out in \cite{Apt81} that the recursion rules themselves are not adequate for proving all correct recursive programs, and they need to be used together with several auxiliary rules. The same is true for recursive quantum circuits (see the examples in Appendices \ref{subsec-qsp} and \ref{subsec-qram} below). The needed auxiliary rules for our purpose of verifying recursive quantum circuits are presented in Table \ref{circuit-rules-proof-4}. Here,  $\mathit{change}(C)$ denotes the set of classical variables modified by $C$ (and the procedure bodies called in $C$). In rule (Invariance), $\mathit{qv}(|\varphi\rangle)\cap\mathit{qv}(C)=\emptyset$ means that $|\varphi\rangle$ is a parameterized quantum state of signature $S$ such that for all classical states $\sigma$, any quantum variables in $\mathit{qv}(C)$ does not appear in $S(\sigma)$.  
\begin{table*}[t]
\begin{equation*}\begin{split}&({\rm Invariance})\ \ \frac{\mathit{free}(A)\cap\mathit{change(C)=\emptyset\qquad \mathit{qv}(|\varphi\rangle)\cap\mathit{qv}(C)=\emptyset}}{\{A, |\varphi\rangle\}\ C\ \{A,|\varphi\rangle\}}\\  
&({\rm Disjunction})\ \ \frac{\left\{A_1,|\varphi\rangle\right\}\ C\ \{B,|\psi\rangle\}\qquad \left\{A_2,|\varphi\rangle\right\}\ C\ \{B,|\psi\rangle\}}{\left\{A_1\vee A_2,|\varphi\rangle\right\}\ C\ \{B,|\psi\rangle\}}\\
&({\rm Conjunction})\ \ \frac{\left\{A_1,|\varphi\rangle\right\}\ C\ \left\{B_1,|\psi\rangle\right\}\qquad \left\{A_2,|\varphi\rangle\right\}\ C\ \left\{B_2,|\psi\rangle\right\}}{\left\{A_1\wedge A_2,|\varphi\rangle\right\}\ C\ \left\{B_1\wedge B_2,|\psi\rangle\right\}}\\
&(\exists{\rm -Introduction})\ \ \frac{\{A,|\varphi\rangle\}\ C\ \{B,|\psi\rangle\}\qquad x\notin\mathit{var}(C)\cup\mathit{free}(B)\qquad \mathit{Inv}(|\psi\rangle:x)}{\{(\exists x)A,|\varphi\rangle\}\ C\ \{B,|\psi\rangle\}}\\
&({\rm Invariance-Con})\ \ \frac{\{A,|\varphi\rangle\}\ C\ \{B,|\psi\rangle\}\qquad \mathit{free}(A^\prime)\cap\mathit{change}(C)=\emptyset}{\{A\wedge A^\prime, |\varphi\rangle\}\ C\ \{B\wedge A^\prime,|\psi\rangle\}}\\
&({\rm Substitution})\ \ \frac{\{A,|\varphi\rangle\}\ C\ \{B,|\psi\rangle\}\qquad \left[\overline{x}\cup\mathit{var}(\overline{t})\right]\cap\mathit{change}(C)=\emptyset\qquad \mathit{Inv}(|\psi\rangle: \overline{x})}{\left\{A[\overline{x}:=\overline{t}],|\varphi\rangle\right\}\ C\ \left\{B[\overline{x}:=\overline{t}],|\psi\rangle\right\}}
\end{split}\end{equation*}
\caption{Auxiliary Rules.}\label{circuit-rules-proof-5}
\end{table*}
It should be pointed out that the soundness of the rule (Conjunction) depends on the assumption that the classical programming language upon which $\mathbf{RQC}^{++}$ is defined is deterministic (see Appendix \ref{proof-auxiliary}).  

\subsection{Proof System} To conclude this section, let us organise the proof rules presented above into two proof (sub)systems: \begin{itemize}\item The \textbf{\textit{proof (sub)system}} $\mathcal{PS}_\mathit{par}$ \textbf{\textit{for partial correctness}} consists of all rules in Tables \ref{circuit-rules-proof-0}, \ref{circuit-rules-proof-1} and \ref{circuit-rules-proof-2}, the rule (Recursion-Par) in Table \ref{circuit-rules-proof-3}, and all rules in Table \ref{circuit-rules-proof-5}. Let $\vdash_\mathit{par}$ stand for the provability in the proof system $\mathcal{PS}_\mathit{par}$.

\item The \textbf{\textit{proof (sub)system}} $\mathcal{PS}_\mathit{tot}$ \textbf{\textit{for total correctness}} consists of all rules in Tables \ref{circuit-rules-proof-0}, \ref{circuit-rules-proof-1}  and \ref{circuit-rules-proof-2}, the rule (Recursion-Tot) in Table \ref{circuit-rules-proof-3}, and all rules in Table \ref{circuit-rules-proof-5} except the rule (Invariance). Let $\vdash_\mathit{tot}$ stand for the provability in the proof system $\mathcal{PS}_\mathit{tot}$.\end{itemize}

\section{Soundness and Completeness Theorems}\label{sec-completeness}\label{sec-sound-complete}

In this section, we establish the soundness and (relative) completeness of the proof system $\mathcal{PS}_\mathit{par}$ for partial correctness and the proof system $\mathcal{PS}_\mathit{tot}$ for total correctness of quantum recursive circuits in $\mathbf{RQC}^{++}$.  

First, we can prove that both $\mathcal{PS}_\mathit{par}$ and $\mathcal{PS}_\mathit{tot}$ are sound:

\begin{thm}[Soundness]\label{soundness}For any quantum circuit $C\in\mathbf{RQC}^{++}$, first-order logical formulas $A,B$ and quantum states $|\varphi\rangle,|\psi\rangle$: \begin{enumerate}\item If $\vdash_\mathit{par}\{A,|\varphi\rangle\}\ C\ \{B,|\psi\rangle\}$, then $\models_\mathit{par}\{A,|\varphi\rangle\}\ C\ \{B,|\psi\rangle\}$; 
\item If $\vdash_\mathit{tot}\{A,|\varphi\rangle\}\ C\ \{B,|\psi\rangle\}$, then $\models_\mathit{tot}\{A,|\varphi\rangle\}\ C\ \{B,|\psi\rangle\}$. 
\end{enumerate}\end{thm}

Now we turn to consider the completeness issue of our logic system for recursive quantum circuits. As is well-known, one can only establish the completeness of a logic for programs including iteration and recursion relative to the theory of an expressive interpretation (see \cite{LS87}, Definition 8.10) of the assertion language. Our recursive definitions of quantum circuits are embedded into a classical programming language. Therefore, what we can expect is also a relative completeness of our logic.   
As shown in \cite{LS87, AB90}, the assertion language $\mathcal{L}$ of the classical programming language used in this paper can be expanded to a first-order language $\mathcal{L}^+$ that enjoys an \textit{arithmetic interpretation} (see \cite{AB90}, Definition 2.10), which contains the standard model of Peano arithmetic and is expressive, and in particular, can define a coding of the finite sequences of elements in the domain.  

Given an interpretation $\mathbb{I}$ of the assertion language. We write $\mathit{Th}(\mathbb{I})=\{A\in\mathcal{L}^+\mid \mathbb{I}\models A\}$ for the theory of $\mathbb{I}$. Furthermore, we use $\mathit{Th}(\mathbb{I}) \vdash_\mathit{par}\{A,|\varphi\rangle\}\ C\ \{B,|\psi\rangle\}$ to denote that the Hoare triple $\{A,|\varphi\rangle\}\ C\ \{B,|\psi\rangle\}$ can be derived from the rules in our proof system $\mathcal{PS}_\mathit{par}$ for partial correctness, together with the assertions in $\mathit{Th}(\mathbb{I})$ as axioms. Likewise, we can define the provability: $\mathit{Th}(\mathbb{I}) \vdash_\mathit{tot}\{A,|\varphi\rangle\}\ C\ \{B,|\psi\rangle\}.$ Then the relative completeness of our proof systems $\mathcal{PS}_\mathit{par}$ and $\mathcal{PS}_\mathit{tot}$ can be stated as the following: 

\begin{thm}[Completeness]\label{completeness} For any arithmetic interpretation $\mathbb{I}$ of the assertion language, we have:
\begin{enumerate}\item  If $\mathbb{I} \models_\mathit{par}\{A,|\varphi\rangle\}\ C\ \{B,|\psi\rangle\}$, then $\mathit{Th}(\mathbb{I}) \vdash_\mathit{par}\{A,|\varphi\rangle\}\ C\ \{B,|\psi\rangle\}$;
\item  If $\mathbb{I} \models_\mathit{tot}\{A,|\varphi\rangle\}\ C\ \{B,|\psi\rangle\}$, then $\mathit{Th}(\mathbb{I}) \vdash_\mathit{tot}\{A,|\varphi\rangle\}\ C\ \{B,|\psi\rangle\}.$
\end{enumerate}
\end{thm}

For better readability, the lengthy and involved proofs of Theorems \ref{soundness} and \ref{completeness} are deferred to Appendix \ref{proof-sound}.

\section{Illustrative Examples}\label{sec-examples}

Now we are going to present a series of examples in order to illustrate how the proof rules developed in Section \ref{proof-system} can be used to verify recursive quantum circuits defined in Section \ref{sec-language}. Duo to the limited space, in this section, we only prove the correctness of quantum Fourier transform $\mathit{QFT}(m,n)$ in Example~\ref{example-QFT} and thus answer Question \ref{ques-1}. Other examples are deferred to Appendix \ref{appendix-examples}, including formal verification of (multi-qubit) controlled gates (Example \ref{exam-control-2}), a quantum circuit generating (multi-qubit) GHZ (Greenberger-Horne-Zeilinger) states, as well as more sophisticated quantum algorithms — quantum state preparation, and quantum random-access memories (QRAMs).

We introduce a classical bit array $j$ of type $\mathbf{integer}\rightarrow\mathbf{Boolean}$.
Slightly abusing the notation, we also use $j[m:n]$ to denote $\sum_{l=m}^n j[l]\cdot 2^{n-l}$
and $0.j[m:n]$ to denote $\sum_{l=m}^n j[l]\cdot 2^{m-l-1}$.
We can define the following parameterized quantum states:
\begin{equation*}
|\Psi\rangle:(j,m,n)\mapsto\frac{1}{\sqrt{2^{n-m+1}}}\bigotimes_{l=1}^{n-m+1}\left(|0\rangle+e^{2\pi i 0.j[n-l+1:n]} |1\rangle\right)\in\hs_2^{\otimes (n-m+1)}.
\end{equation*} 
Then the correctness of $\mathit{QFT}(m,n)$ can be specified by the following judgement:
\begin{equation}\label{QFT-correct}m\leq n\models_\mathit{tot}\left\{|j[m:n]\rangle\right\}\ \mathit{QFT}(m,n)\ \left\{|\Psi\rangle(j,m,n)\right\},
\end{equation}
where $|j[m:n]\rangle$ and $|\Psi\rangle(j,m,n)$ are two parameterized state with explicit parameters $j,m,n$.

In the following, we verify (\ref{QFT-correct}) using the proof rules presented in Section \ref{proof-system}. For a better understanding, 
let us set a parameterized quantum state
\begin{equation*}
    |J'\rangle : (j,m,n,\theta) \mapsto (S_{(j[m+1:n]+\theta)/2^{n-m}}|j[m]\rangle)|j[m+1:n]\rangle\in\hs_2^{\otimes (n-m+1)},
\end{equation*}
where $S_\theta$ is defined in (\ref{eq:gate_S}).
 To apply the rule (Recursion-Tot), 
we define ranking function $r=n-m$ and introduce the assumptions:
\begin{align}
    \label{eq:QFT-ass}
    \vdash_{\mathit{tot}}&\left\{m\leq n\wedge r<z, |j[m:n]\rangle\right\} \mathit{QFT}(m,n)\left\{\mathit{true}, |\Psi\rangle(j,m,n)\right\},\\
    \label{eq:Rot-ass}
    \vdash_{\mathit{tot}}&\left\{m\leq n \wedge r<z, |j[m:n]\rangle\right\}
    \mathit{Rot}(m,n,\theta)\left\{\mathit{true}, |J'\rangle(j,m,n,\theta)\right\},\\
    \label{eq:Shift-ass}
    \vdash_{\mathit{tot}}&\left\{m\leq n\wedge r<z, |j[m:n]\rangle\right\} \mathit{Shift}(m,n)\left\{\mathit{true}, |j[n]\rangle |j[m:n-1]\rangle\right\},
\end{align}
where $z$ is a fresh variable.
Note that the quantum states in (\ref{eq:QFT-ass}), (\ref{eq:Rot-ass}) and (\ref{eq:Shift-ass}) (as well as in the rest of the proof)
are all parameterized.
Since $m\leq n\rightarrow r\geq 0$ is obviously true,
with the rule (Block) and certain substitution of variables, it suffices to prove that
\begin{align}
    \label{eq:QFT-goal}
    \vdash_{\mathit{tot}}&\left\{m\leq n \wedge r=z, |j[m:n]\rangle\right\}\mathbf{if}\ldots \mathit{QFT}\ldots\mathbf{fi}\left\{\mathit{true}, |\Psi\rangle(j,m,n)\right\},\\
    \label{eq:Rot-goal}
    \vdash_{\mathit{tot}}&\left\{m\leq n\wedge r=z, |j[m:n]\rangle\right\}\mathbf{if}\ldots \mathit{Rot}\ldots \mathbf{fi}\left\{\mathit{true}, |J'\rangle (j,m,n,\theta)\right\},\\
    \label{eq:Shift-goal}
    \vdash_{\mathit{tot}}&\left\{m\leq n\wedge r=z, |j[m:n]\rangle\right\}\mathbf{if}\ldots \mathit{Shift}\ldots \mathbf{fi}\left\{\mathit{true},|j[n]\rangle |j[m:n-1]\rangle\right\},
\end{align}
where $\mathbf{if}\ldots \mathit{QFT}\ldots \mathbf{fi}$, $\mathbf{if}\ldots \mathit{Shift}\ldots \mathbf{fi}$ and $\mathbf{if}\ldots \mathit{Rot}\ldots \mathbf{fi}$
stand for the right hand sides of the programs in (\ref{p-QFT-1}), (\ref{p-QFT-m}) and (\ref{p-QFT-2}), respectively.

\textbf{Step I}: 
Using the rules (GA-P), (SK-P) and (Consequence), we can prove: 
\begin{align}
    \label{eq:QFT-base}
    \vdash_{\mathit{tot}}&\left\{m= n \wedge r=z, |j[m:n]\rangle\right\}S_{0}[q[m]]\left\{\mathit{true}, |\Psi\rangle(j,m,n)\right\},\\
    \label{eq:Rot-base}
    \vdash_{\mathit{tot}}&\left\{m= n \wedge r=z, |j[m:n]\rangle\right\}S_{\theta}[q[m]]\left\{\mathit{true}, |J'\rangle (j,m,n,\theta)\right\},\\
    \label{eq:Shift-base}
    \vdash_{\mathit{tot}}&\left\{m= n\wedge r=z, |j[m:n]\rangle\right\}\mathbf{skip}\left\{\mathit{true}, |j[n]\rangle|j[m:n-1]\rangle\right\}.
\end{align}

\textbf{Step II}: 
Note that $m<n\models m+1\leq n$, $m<n\models m\leq n-1$ and $r=z\models r-1<z$.
Applying the rules (Instantiation) and (Consequence) to assumptions (\ref{eq:QFT-ass}), (\ref{eq:Rot-ass}) and (\ref{eq:Shift-ass}), we obtain:
\begin{align}
    \label{eq:QFT-ins}
    \vdash_{\mathit{tot}}&\left\{m< n\wedge r=z, |j[m+1:n]\rangle\right\} \mathit{QFT}(m+1,n)\left\{\mathit{true}, |\Psi\rangle(j,m+1,n)\right\},\\
    \label{eq:Rot-ins}
    \vdash_{\mathit{tot}}&\left\{m< n\wedge r=z, |j[m:n-1]\rangle\right\} \mathit{Rot}(m,n-1, \theta')\left\{\mathit{true}, |J'\rangle(j,m,n-1,\theta')\right\},\\
    \label{eq:Shift-ins}
    \vdash_{\mathit{tot}}&\left\{m< n\wedge r=z, |j[m+1:n]\rangle\right\}\mathit{Shift}(m+1,n)\left\{\mathit{true}, |j[n]\rangle |j[m+1:n-1]\rangle\right\},
\end{align}
where $\theta'=(\theta+j[n])/2$.
Note that $\mathit{free}(m<n \wedge r=z)\cap \mathit{change}\left(\mathit{QFT}(m+1,n)\right)=\emptyset$. 
Using the rule (Invariance-Con) we can strengthen (\ref{eq:QFT-ins}) to
\begin{align}
    \label{eq:QFT-ins-s}
    \vdash_{\mathit{tot}}\left\{m<n\wedge r=z, |j[m+1:n]\rangle\right\} \mathit{QFT}(m+1,n)\left\{m< n \wedge r=z, |\Psi\rangle(j,m+1,n)\right\}.
\end{align}
Similarly, using the rules (Instantiation), (Consequence) and (Invariance-Con) on (\ref{eq:Rot-ass}), we have
\begin{equation}
    \label{eq:Rot-ins-s}
    \vdash_{\mathit{tot}}\left\{m< n\wedge r=z, |j[m:n]\rangle\right\} \mathit{Rot}(m,n, 0)\left\{m< n \wedge r=z, |J'\rangle(j,m,n,0)\right\}.
\end{equation}

\textbf{Step III}:
Let us first derive (\ref{eq:Shift-goal}).
From the rule (GA-P), we have: 
\begin{equation}
    \label{eq:Shift-swap}
    \begin{split}
        \vdash_{\mathit{tot}}&\left\{m<n\wedge r=z, |j[m:n]\rangle\right\}\mathit{Swap}[q[m],q[n]]\\
        &\qquad\qquad\left\{m<n\wedge r=z, |j[n]\rangle|j[m+1:n-1]\rangle |j[m]\rangle\right\}.
    \end{split}
\end{equation}
By the rules (SC-P), (Frame) and (Consequence) we can combine (\ref{eq:Shift-ins}) and (\ref{eq:Shift-swap}) to obtain: 
\begin{equation}
    \label{eq:Shift-in-if}
    \begin{split}
        \vdash_{\mathit{tot}}&\left\{m< n\wedge r=z, |j[m:n]\rangle\right\}\\
        &\qquad\qquad\mathit{Swap}[q[m],q[n]];\mathit{Shift}(m+1,n)\left\{\mathit{true}, |j[n]\rangle |j[m:n-1]\rangle\right\}.
    \end{split}
\end{equation}
Then applying the rule (IF-P) to (\ref{eq:Shift-base}) and (\ref{eq:Shift-in-if}) yields (\ref{eq:Shift-goal}).

\textbf{Step IV}:
Similarly, let us derive (\ref{eq:Rot-goal}).
Using the rules (QIF-P), (Invariance-Con) and (Consequence) on (\ref{eq:Rot-ins}), we have
\begin{equation}
    \label{eq:Rot-qif}
    \begin{split}
        \vdash_{\mathit{tot}}&\left\{m< n\wedge r=z, |j[m:n]\rangle\right\}\mathbf{qif}\ldots\mathbf{fiq}\left\{\mathit{true}, |J'\rangle (j,m,n,\theta)\right\},
    \end{split}
\end{equation}
where 
$\mathbf{qif}\ldots\mathbf{fiq}$ denotes the segment of program (\ref{p-QFT-m}) between $\mathbf{qif}$ and $\mathbf{fiq}$,
and we use $ |J'\rangle (j,m,n-1,\theta')|j[n]\rangle = |J'\rangle (j,m,n,\theta)$.
Applying the rule (IF-P) to (\ref{eq:Rot-base}) and (\ref{eq:Rot-qif}) yields (\ref{eq:Rot-goal}).

\textbf{Step V}:
Now let us derive (\ref{eq:QFT-goal}).
Note that $|J'\rangle (j,m,n,0)=\frac{1}{\sqrt{2}}(|0\rangle + e^{2\pi i 0.j[m:n]}|1\rangle) |j[m+1:n]\rangle$.
We can then apply the rules (SC-P), (Frame) and (Consequence) to combine (\ref{eq:Rot-ins-s}), (\ref{eq:QFT-ins-s}) and (\ref{eq:Shift-ass}) to obtain:
 \begin{equation}
    \label{eq:QFT-else}
    \begin{split}
        \vdash_{\mathit{tot}}\left\{m< n\wedge r=z, |j[m:n]\rangle\right\}\mathit{Rot}(m,n,0);\mathit{QFT}(m+1,n);\mathit{Shift}(m,n)\left\{\mathit{true}, |\Psi\rangle(j,m,n)\right\}.
    \end{split}
\end{equation}
Hence, finally applying the rule (IF-P) to (\ref{eq:QFT-base}) and (\ref{eq:QFT-else}) gives (\ref{eq:QFT-goal}).
 
\section{Conclusion}

In this paper, we defined a programming language $\mathbf{RQC}^{++}$ for recursively defined quantum circuits by employing quantum if-statements \cite{Alt05, Ying12, Sab18, Yuan, Bich, Voi} and introducing procedures parameterised by classical expressions. This language can serve as a vehicle in which a quantum Divide-and-Conquer strategy (not classical Divide-and-Conquer strategy used in the design of quantum algorithms) can be properly programmed (see Example \ref{exam-qram}).

Furthermore, we presented a sound and (relatively) complete logic system for reasoning about  quantum circuits and algorithms programmed in  $\mathbf{RQC}^{++}$. It provides for the first time a principled method for verification of recursively defined quantum circuits. A series of application examples of our verification method are shown, including multi-qubit controlled quantum gates, quantum circuits generating the GHZ states, quantum Fourier transform, quantum state preparation and quantum random access memories (QRAM) (More sophisticated examples of quantum algorithms using the quantum Divide-and-Conquer strategy and their verification will be presented in a forthcoming paper).

\textbf{Topics for Future Research}: As always happened in all verification techniques for both classical and quantum programs, formal verifications of these examples are tedious, although they were carried out in a principled way. So, a natural topic for the research of next step is to mechanise the verification by implementing the logic system of this paper in the existing verification frameworks and tools for (i) quantum information and computing in general (e.g. Quantomatic \cite{KZ15}, ZX-calculus \cite{Bob11, KW19}); and (ii) for quantum programs particularly (e.g. QWire \cite{Paykin17}, QBricks \cite{Char21}, VOQC \cite{Rand21}, Gleipnir \cite{Tao21}, QHLProver \cite{Liu19} and CoqQ \cite{Zhou23}). Classical recursive programming usually employs rich data structures. The notion of inductive data structure was recently introduced in \cite{Zam20} into quantum programming. So, another interesting topic for future research is how the scheme of quantum recursion defined in this paper can be empowered by inductive quantum data structures and others.

\newpage

\appendix

\section{More Examples}\label{appendix-examples}

In this appendix, we show more applications of our proof system in verification of quantum recursive programs.  

We first note that the rules (Recursion-Par) and (Recursion-Tot) were presented in Table \ref{circuit-rules-proof-3} in a general form where multiple procedure identifiers are allowed to be mutually recursively defined. 
The QFT example given in Section \ref{sec-examples} has more than one procedure identifier. 
However, the examples that we consider in this appendix include only a single procedure identifier. In this case, the rules (Recursion-Par) and (Recursion-Tot) can be simplified as (Recursion-Par-S) and (Recursion-Tot-S), respectively,  in Table \ref{circuit-rules-proof-4}. 
\begin{table*}[t]
\begin{equation*}\begin{split}&({\rm Recursion-Par-S})\ \ \frac{\{A,|\varphi\rangle\}\ P(\overline{t})\ \{B,|\psi\rangle\}\ \vdash\ \{A,|\varphi\rangle\}\ \mathbf{begin\ local}\ \overline{u}:=\overline{t};C\ \mathbf{end}\ \{B,|\psi\rangle\}}{\{A, |\varphi\rangle\}\ P(\overline{t})\ \{B,|\psi\rangle\}}\\ \\ 
&({\rm Recursion-Tot-S})\ \ \frac{\begin{split} &\{A\wedge r<z,|\varphi\rangle\}\ P(\overline{t})\ \{B,|\psi\rangle\}\ \vdash\ \{A\wedge r=z,|\varphi\rangle\}\\ &\qquad\qquad\qquad\qquad\qquad\qquad\qquad\mathbf{begin\ local}\ \overline{u}:=\overline{t};C\ \mathbf{end}\ \{B,|\psi\rangle\}\\
&A\rightarrow r\geq 0
\end{split}}{\{A, |\varphi\rangle\}\ P(\overline{t})\ \{B,|\psi\rangle\}}
\end{split}\end{equation*}
\caption{Rule for Simple Recursion. In rule (Recursion-Tot-S), there is only one recursive procedure declared by $\mathcal{D}=P(\overline{u})\Leftarrow C$, $r$ is an integer expression, and $z\notin\mathit{var}(r)\cup \mathit{free}(A)\cup\mathit{free}(B)\cup \mathit{var}(C)$ is an integer variable.} \label{circuit-rules-proof-4}
\end{table*}

We will sometimes adopt the \textit{decomposition strategy} widely practiced in verification of classical recursive programs: The total correctness of a program can be proved by proving its partial correctness and its termination separately. In our case of recursive quantum circuits, we can prove the total correctness: $$\vdash_\mathit{tot}\{A,|\varphi\rangle\}\ C\ \{B,|\psi\rangle\}$$ by proving the following two judgements: \begin{enumerate}\item \textit{Partial correctness}: $\vdash_\mathit{par}\{A,|\varphi\rangle\}\ C\ \{B,|\psi\rangle\}$; and
\item \textit{Termination}: $\vdash_\mathit{tot}\{A,|\varphi\rangle\}\ C\ \{\mathit{true},|\theta\rangle\}$ for some $|\theta\rangle$. 
\end{enumerate} 

\subsection{Controlled Gates}

In this subsection, we verify the correctness of controlled gates, i.e. program (\ref{c-gate}) in Example \ref{exam-control-1}, using the proof rules presented in Section \ref{proof-system}. 
We only prove the Hoare triple (\ref{c-gate-1}), and (\ref{c-gate-2}) can be proved similarly. 
Recall that our goal is to prove
\begin{equation*}
    (m\leq n)\wedge(\forall m\leq k<n)(i[k]=1)\models_\mathit{tot} \left\{|\Phi\rangle(i,m,n)\right\}\ C^{(\ast)}(m,n)\ \left\{|\Psi\rangle(i,m,n)\right\},
\end{equation*}
where $|\Phi\rangle(i,m,n)$ and $|\Psi\rangle(i,m,n)$ are parameterized quantum states defined by:
\begin{equation*}
|\Phi\rangle:(i,m,n)\mapsto |i[m:n]\rangle,\qquad\qquad
|\Psi\rangle:(i,m,n)\mapsto |i[m:n-1]\rangle U|i[n]\rangle.
\end{equation*}
For simplicity of notation, we use $A$ to denote $(\forall m\leq k<n)(i[k]=1)$.
Let us take $r=m-n$ as the ranking function. Since $m\leq n\rightarrow r\geq 0$ is obviously true, by the rules (Recursion-Tot-S) and (Block) (with certain substitution to avoid variable conflict), it suffices to show that 
\begin{equation}
\label{c-gate-premise}
\begin{split}
\left\{m\leq n\wedge r<z\wedge A,|\Phi\rangle(i,m,n)\right\}\ C^{(\ast)}(m,n)\ &\left\{\mathit{true}, |\Psi\rangle(i,m,n)\right\}\vdash\\ \left\{m\leq n\wedge r=z\wedge A,|\Phi\rangle(i,m,n)\right\}&\ \mathbf{if} ... \mathbf{fi}\ \left\{\mathit{true}, |\Psi\rangle(i,m,n)\right\}
\end{split}
\end{equation} 
where $\mathbf{if}...\mathbf{fi}$ stands for the right-hand side of (\ref{c-gate}).

Now we prove the logical consequence in equation (\ref{c-gate-premise}). Assume that 
\begin{equation}\label{assume-control}\vdash_\mathit{tot} \left\{m\leq n\wedge r<z\wedge A,|\Phi\rangle(i,m,n)\right\}\ C^{(\ast)}(U)(m,n)\ \left\{\mathit{true},|\Psi\rangle(i,m,n)\right\}. 
\end{equation} 
 
\textbf{Step I}: It follows from the rue (GA-P) that 
\begin{equation}\label{ga-p}\vdash_\mathit{tot} \left\{m=n,|i[n]\rangle\right\}\ U[q[n]]\ \left\{m=n, U|i[n]\rangle\right\}.
\end{equation} Directly by definition, we have:
\begin{equation}\label{entail}\left\{m=n\wedge r=z,|\Phi\rangle(i,m,n)\right\}\models\left\{m=n,|i[n]\rangle\right\},\qquad \left\{m=n,U|i[n]\rangle\right\}\models\left\{\mathit{true},|\Psi\rangle(i,m,n)\right\}.
\end{equation} Using the rule (Consequence) on (\ref{ga-p}) and (\ref{entail}), we obtain:
\begin{equation}\label{ga-p1}\vdash_\mathit{tot}\left\{m=n\wedge r=z,|\Phi\rangle(i,m,n)\right\}\ U[q[n]]\ \left\{\mathit{true},|\Psi\rangle(i,m,n)\right\}.
\end{equation}

\textbf{Step II}: By the rule (Instantiation) on the assumption (\ref{assume-control}) we have: 
\begin{equation}\label{}\vdash_\mathit{tot} \left\{m+1\leq n\wedge n-(m+1)<z\wedge A',|\Phi\rangle(i,m+1,n)\right\}\ C^{(\ast)}(U)(m+1,n)\ \left\{\mathit{true},|\Psi\rangle(i,m+1,n)\right\},
\end{equation} 
where $A'$ denotes $(\forall m+1\leq k<n)(i[k]=1)$.
By the rules (QIF-P), (Consequence) and noticing that $$m<n\wedge r=z\wedge A'\models m+1\leq n\wedge n-(m+1)<z\wedge A',$$ we obtain: 
\begin{equation}\label{qif-a}\vdash_\mathit{tot}\left\{m<n\wedge r=z\wedge A,|1\rangle|\Phi\rangle(i,m+1,n)\right\}\ \mathbf{qif}...\mathbf{fiq}\ \left\{\mathit{true},|1\rangle|\Psi\rangle(i,m+1,n)\right\},\end{equation}
where $\mathbf{qif}...\mathbf{fiq}$ stands for the segment of (\ref{c-gate}) between $\mathbf{qif}$ and $\mathbf{fiq}$. 
As $\mathit{free}(A)\cap \mathit{change}(\mathbf{qif}\ldots\mathbf{fiq})=\emptyset$,
using (Invariance-Con) we can strengthen (\ref{qif-a}) to 
\begin{equation}
    \vdash_\mathit{tot}\left\{m<n\wedge r=z\wedge A,|1\rangle|\Phi\rangle(i,m+1,n)\right\}\ \mathbf{qif}...\mathbf{fiq}\ \left\{A,|1\rangle|\Psi\rangle(i,m+1,n)\right\},
\end{equation}
which further implies
\begin{equation}
    \label{eq:controlled-qif-b}
    \vdash_\mathit{tot}\left\{m<n\wedge r=z\wedge A,|\Phi\rangle(i,m,n)\right\}\ \mathbf{qif}...\mathbf{fiq}\ \left\{\mathit{true},|\Psi\rangle(i,m,n)\right\}
\end{equation}
by the rule (Consequence) and noting that $A\models |\Phi\rangle(i,m,n)=|1\rangle|\Phi\rangle(i,m+1,n)$ and $A\models |\Psi\rangle(i,m,n)=|1\rangle|\Psi\rangle(i,m+1,n)$.

\textbf{Step III}: Using the rule (IF) we derive:
\begin{equation}\label{if-a}\vdash_\mathit{tot}\left\{m\leq n\wedge r=z,|\Phi\rangle(i,m,n)\right\}\ \mathbf{if} ... \mathbf{fi}\ \left\{\mathit{true},|\Psi\rangle(i,m,n)\right\}\end{equation} from (\ref{ga-p1}) and (\ref{eq:controlled-qif-b}). Therefore, the proof is completed.

\subsection{GHZ (Greenberger-Horne-Zeilinger) States} 

The GHZ states are an important type of entangled quantum states with certain extremely non-classical properties. A recursively defined quantum circuit that generates the GHZ states was first given in \cite{Deng24}. This quantum recursive circuit and its correctness can be described in the language introduced in this paper as follows: 

\begin{exam} Recall that the GHZ state on the section $q[m,n]$ of qubits is a parameterized quantum state defined by: 
\begin{equation*}
|\mathit{GHZ}\rangle: (m,n)\mapsto \frac{1}{\sqrt{2}}\left(|0\rangle^{\otimes(n-m+1)}+|1\rangle^{\otimes(n-m+1)}\right).
\end{equation*}
A quantum circuit $U(m,n)$ that generates the GHZ state from initial state $|0\rangle^{\otimes(n-m+1)}$ can be recursively defined by
\begin{equation}\label{def-GHZ}\begin{split}U(m,n)\Leftarrow\ &\mathbf{if}\ m=n\ \mathbf{then}\ H[q[n]] \ \mathbf{else}\ U(m,n-1);\mathit{CNOT}[q[n-1],q[n]]\ \mathbf{fi}
\end{split}\end{equation}  The correctness of this circuit can be specified by the following Hoare triple:
\begin{equation}\label{GHZ-correct}m\leq n\models_\mathit{tot}\left\{|0\rangle^{\otimes(n-m+1)}\right\}\ U(m,n)\ \left\{|\mathit{GHZ}\rangle(m,n)\right\}.\end{equation}\end{exam}

The correctness of the quantum recursive circuit generating the GHZ states was specified in \cite{Deng24} as a hypothesis-amplitude judgement defined there. Then it was verified using some proof rules for hypothesis-amplitude judgements based on the path-sum representation of quantum circuits (i.e. a discrete variant of the Feynman path integral) \cite{Amy18} with the help of an SMT solver. Such a verification is not purely syntax-oriented, and it is carried out essentially at the level of semantics. 

Here, we present a much simpler verification of (\ref{GHZ-correct}) only using the proof rules given in Section~\ref{proof-system}. Indeed, the verification is similar to that for controlled gates given in the previous subsection. 
We choose ranking function $r=n-m$. Obvious, $m\leq n\rightarrow r\geq 0$ is true. Thus by the rules (Recursion-Tot-S) and (Block) (with certain substitutions to avoid variable conflict), we only need to prove:
\begin{equation}\label{GHZ-premise}\begin{split}&\left\{m\leq n\wedge r<z, |0\rangle^{\otimes(n-m+1)}\right\}\ U(m,n)\ \{\mathit{true},|\mathit{GHZ}\rangle(m,n)\}\\ &\qquad\qquad\qquad\qquad\qquad\vdash\left\{m\leq n\wedge r=z,|0\rangle^{\otimes(n-m+1)}\right\}\ \mathbf{if} ...\mathbf{fi}\ \{\mathit{true},|\mathit{GHZ}\rangle(m,n)\},\end{split}\end{equation} where $\mathbf{if}...\mathbf{fi}$ stands for the program on the right-hand side of (\ref{def-GHZ}).

As always in applying the proof rule for recursion, let us introduce the assumption:
\begin{equation}\label{GHZ-correct-0}\vdash_\mathit{tot}\left\{m\leq n\wedge r<z, |0\rangle^{\otimes(n-m+1)}\right\}\ U(m,n)\ \left\{\mathit{true}, |\mathit{GHZ}\rangle(m,n)\right\}.\end{equation}   
Then we prove in the following that
\begin{equation}\label{GHZ-correct-1}\vdash_\mathit{tot}\left\{m\leq n\wedge r=z, |0\rangle^{\otimes(n-m+1)}\right\}\ \mathbf{if} ...\mathbf{fi}\ \left\{\mathit{true}, |\mathit{GHZ}\rangle(m,n)\right\}.\end{equation} 

\textbf{Step I}: By the rules (GA-P) and (Consequence), we can prove:
\begin{equation}\label{GHZ-correct-2}\vdash_\mathit{tot} \left\{m=n\wedge r=z,|0\rangle^{\otimes(n-m+1)}\right\}\ H[q[n]]\ \{\mathit{true},|\mathit{GHZ}\rangle(m,n)\}.\end{equation}

\textbf{Step II}: Applying the rule (Instantiation) to assumption (\ref{GHZ-correct-0}), together with the rules (Frame) and (GA-P), yields:
\begin{equation*}\begin{split}\vdash &\left\{m\leq n-1\wedge (n-1)-m<z,|0\rangle^{\otimes(n-m+1)}\right\}\ U(m,n-1)\ \{\mathit{true},|\mathit{GHZ}\rangle(m,n-1)\otimes |0\rangle\}\\ &\qquad \qquad \qquad\   \qquad\qquad\qquad\qquad\qquad\qquad\mathit{CNOT}[q[n-1],q[n]]\ \{\mathit{true}, |\mathit{GHZ}\rangle(m,n)\}.\end{split}\end{equation*} Note that $m<n\wedge r=z\models m\leq n-1\wedge (n-1)-m<z$. Then by the rule (Consequence) we obtain: 
\begin{equation}\label{GHZ-correct-3}\left\{m<n\wedge r=z,|0\rangle^{\otimes(n-m+1)}\right\}\ U(m,n-1);\mathit{CNOT}[q[n-1],q[n]]\ \{\mathit{true}, |\mathit{GHZ}\rangle(m,n)\}.\end{equation}

\textbf{Step III}: Now we use the rule (IF-P) to combine (\ref{GHZ-correct-2}) and (\ref{GHZ-correct-3}) and obtain equation (\ref{GHZ-correct-1}). Thus, the proof is completed.

\subsection{Quantum State Preparation}\label{subsec-qsp}

Quantum state preparation (QSP) is a common subroutine in many quantum algorithms, including quantum simulation and quantum machine learning.
The task is to generate the $n$-qubit state:
$$\frac{1}{\sqrt{a}}\sum_{j=0}^{N-1}\sqrt{a_j}|j\rangle$$ from state $|0\rangle^n$,
with $N=2^n$ and $a=\sum_{j=0}^{N-1}|a_j|$,
given the $N$-dimensional complex vector $\left(a_j\right)_{j=0}^{N-1}\in\mathbb{C}^N$. Using our programming language defined in Section \ref{sec-language}, a  quantum circuit for the quantum state preparation (see for example \cite{KP17}) can be written as the following quantum recursive program.

\begin{exam}
Let $b$ and $\theta$ be classical arrays of type $\mathbf{integer}\times \mathbf{integer}\rightarrow \mathbf{float}$,
such that for each $0\leq j<N$, $b[n,j]=|a_j|$ and $a_j=e^{i\theta[j]}b[j]$.
Let $S$ be a classical array of type $\mathbf{integer}\times\mathbf{integer}\times\mathbf{integer}\rightarrow \mathbf{float}$, such that for any $0\leq l<r\leq N$, and $S[n,l,r]=\sum_{j=l}^{r-1}b[n,j]$.
Let $\gamma$ and $\beta$ classical arrays of type $\mathbf{integer}\times \mathbf{integer}\times\mathbf{integer}\rightarrow \mathbf{float}$, such that for any $0\leq k < n $ and $0\leq x< 2^k$, $\gamma[n,k,x]=\frac{S[n,u,w]}{S[n,u,v]}$, $\beta[n,k,x]=\theta[n,w]-\theta[n,u]$,
where $u=2^{n-k}x$, $v=2^{n-k}x+2^{n-k}$, and $w=\frac{u+v}{2}$.
Define a single qubit gate $U(n,k,x)$ by 
$$U(n,k,x)|0\rangle=\sqrt{\gamma[n,k,x]}|0\rangle+e^{i\beta[n,k,x]/2}\sqrt{1-\gamma[n,k,x]}|1\rangle,$$
Let 
\begin{equation}
\label{eq:QSP-prog}
\begin{split}
\mathit{QSP}(n,k,x)\Leftarrow
\ 
&\mathbf{if}\ k< n\ \mathbf{then}\\
&\ \ U(n,k,x)[q[k+1]];\\
&\ \ \mathbf{qif}[q[k]](\square_{y=0}^{1} |y\rangle\rightarrow\mathit{QSP}(n,k+1,2x+y))\ \mathbf{fiq}\\
&\mathbf{fi}.
\end{split}
\end{equation} 
Then one calls $\mathit{QSP}(n,0,0)$ for the quantum state preparation. 

Consider any fixed integer $n$. We define quantum states: 
\begin{equation*}
    |\psi\rangle : (n,b,\theta,l,r)\mapsto \sum_{j=l}^{r-1} \frac{1}{\sqrt{S[n,l,r]}} e^{i\theta[n,j]/2}\sqrt{b[n,j]}| j\rangle.
\end{equation*}

For simplicity, let us use $A$ to denote $k<n\wedge 0\leq x< 2^k$,
and set the parameterized quantum state:
\begin{equation*}
    |\phi\rangle:(n,b,\theta,k,x)\mapsto |\psi\rangle (n,b,\theta,2^{n-k}x,2^{n-k}x+2^{n-k}).
\end{equation*}
Then the correctness of the program in (\ref{eq:QSP-prog}) can be expressed as the following Hoare triple:
\begin{equation}
     \label{eq:QSP-k-1-n}
  A \models_{\mathit{tot}}\left\{|0\rangle^{\otimes (n-k)}\right\}\mathit{QSP}(n,k,x)
    \left\{|\phi\rangle(n,b,\theta,k,x)\right\}.
\end{equation}
\end{exam}  

Now let us verify (\ref{eq:QSP-k-1-n}) using the proof rules presented in Section \ref{proof-system}. To apply the rule (Recursion-Tot-S), we define ranking function $r=n-k$ and introduce the assumption:
\begin{equation}
    \label{eq:QSP-ass}
    \vdash_{\mathit{tot}}\left\{A\wedge r<z , |0\rangle^{\otimes (n-k)}\right\}\mathit{QSP}(n,k,x)
    \left\{\mathit{true}, |\phi\rangle(n,b,\theta,k,x)\right\}.
\end{equation}
Since $k<n\rightarrow r\geq 0$ is obviously true,
with the rule (Block) and certain substitution of variables, it suffices to prove that
\begin{equation}
    \label{eq:QSP-goal}
    \vdash_{\mathit{tot}}\left\{A\wedge r=z, |0\rangle^{\otimes (n-k)}\right\}
    \ \mathbf{if}\ldots\mathbf{fi}
    \left\{\mathit{true}, |\phi\rangle(n,b,\theta,k,x)\right\},
\end{equation}
where $\mathbf{if}\ldots\mathbf{fi}$ stands for the right hand side of the program in (\ref{eq:QSP-prog}).


\textbf{Step I}: Applying the rules (Instantiation) and (Consequence) to assumption (\ref{eq:QSP-ass}), we obtain:
\begin{equation}
    \begin{split}
    &\vdash_{\mathit{tot}}\left\{A'\wedge n-(k+1)<z,|0\rangle^{\otimes (n-k-1)}\right\}\\
    &\qquad\qquad \mathit{QSP}(n,k+1,2x+y)
    \left\{\mathit{true}, |\phi\rangle(n,b,\theta, k+1,2x+y)\right\},
    \end{split}
\end{equation} 
where $A'$ denotes $k+1< n\wedge 0\leq y \leq 1\wedge 0\leq x< 2^{k}$.
Since $A'\wedge n-k=z\models A'\wedge n-(k+1)<z$, by the rule (Consequence) it holds that  
\begin{equation}
    \label{eq:QSP-k+1}
    \vdash_{\mathit{tot}}\left\{A'\wedge r=z,|0\rangle^{\otimes (n-k-1)}\right\}\mathit{QSP}(n,b,\theta,k+1,2x+y)
    \left\{\mathit{true}, |\phi\rangle(n,b,\theta, k+1,2x+y)\right\}.
\end{equation}


\textbf{Step II}: Let us use the rule (QIF-P) on (\ref{eq:QSP-k+1}). Then it yields:
\begin{equation}
    \label{eq:QSP-qif}
        \vdash_{\mathit{tot}}\left\{A'\wedge r=z , |y\rangle |0\rangle^{\otimes (n-k-1)}\right\}\mathbf{qif}\ldots\mathbf{fiq}\left\{\mathit{true}, |y\rangle|\phi\rangle(n,b,\theta, k+1, 2x+y)\right\},
\end{equation}
where $\mathbf{qif}\ldots\mathbf{fiq}$ stands for the corresponding part in the program in (\ref{eq:QSP-prog}).

From the rule (GA-P), we can derive:
\begin{equation}
    \label{eq:QSP-rot}
    \begin{split}
        \vdash_{\mathit{tot}}&\left\{k+1< n\wedge r=z\wedge 0\leq x< 2^{k}, |0\rangle\right\} U(n,k,x)[q[k+1]]\\
        &\qquad\left\{k+1< n\wedge r=z\wedge 0\leq x< 2^{k},\sqrt{\gamma[n,k,x]}|0\rangle+e^{i\beta[n,k,x]/2}\sqrt{1-\gamma[n,k,x]}|1\rangle\right\}.
    \end{split}
\end{equation}
Applying the rules (SC-P), (Consequence), (Frame) and (Linearity) to combine (\ref{eq:QSP-rot}) and (\ref{eq:QSP-qif}) to obtain:
\begin{equation}
    \label{eq:QSP-sc}
    \begin{split}
        \vdash_{\mathit{tot}}&\left\{k+1< n\wedge 0\leq x< 2^{k}\wedge r=z, |0\rangle^{\otimes (n-k)}\right\}\\ 
         &\qquad U(n,k,x)[q[k+1]];\ \mathbf{qif}\ldots\mathbf{fiq}
         \left\{\mathit{true}, |\phi\rangle(n,b,\theta,k,x)\right\},
    \end{split}
\end{equation}
where we use 
\begin{equation*}
    \begin{split}
         &k+1< n\wedge 0\leq x< 2^{k}\wedge r=z\models\sqrt{\gamma[n,k,x]}|0\rangle|\phi\rangle(n,b,\theta,k+1,2x)+\\
         &\qquad e^{i\beta[n,k,x]/2}\sqrt{1-\gamma[n,k,x]}|1\rangle|\phi\rangle(n,b,\theta,k+1,2x+1)=
         |\phi\rangle(n,b,\theta,k,x).
    \end{split}
\end{equation*}


\textbf{Step III}:
Now consider the base case. 
By the rule (SK-P), we have for any $|\eta\rangle$:
\begin{equation*}
    \vdash_{\mathit{tot}}\left\{k=n\wedge r=z, |\eta\rangle\right\} \mathbf{skip}\left\{k=n\wedge r=z, |\eta\rangle\right\},
\end{equation*}
which, together with the rule (IF-P), further implies: 
\begin{equation}
    \vdash_{\mathit{tot}}\left\{k=n\wedge r=z, |\eta\rangle\right\} \mathit{QSP}(k,n,x)\left\{k=n\wedge r=z, |\eta\rangle\right\}.
\end{equation}
Following some steps similar to the above, we can derive:
\begin{equation}
    \begin{split}
        \vdash_{\mathit{tot}}&\left\{k=n-1\wedge 0\leq x< 2^{k}\wedge r=z, |0\rangle^{\otimes (n-k)}\right\}\\ 
         &\qquad U(n,k,x)[q[k+1]];\ \mathbf{qif}\ldots\mathbf{fiq}
         \left\{\mathit{true}, |\phi\rangle(n,b,\theta,k,x)\right\},
    \end{split}
\end{equation}
which combined with (\ref{eq:QSP-sc}), by the rule (Disjunction), yields:
\begin{equation}
    \label{eq:QSP-sc-p}
        \begin{split}
        \vdash_{\mathit{tot}}&\left\{k< n\wedge 0\leq x< 2^{k}\wedge r=z, |0\rangle^{\otimes (n-k)}\right\}\\ 
         &\qquad U(n,k,x)[q[k+1]];\ \mathbf{qif}\ldots\mathbf{fiq}
         \left\{\mathit{true}, |\phi\rangle(n,b,\theta,k,x)\right\}.
    \end{split}
\end{equation}

Finally, using the rules (IF-P) and (SK-P) on (\ref{eq:QSP-sc-p}), we obtain (\ref{eq:QSP-goal}).

\subsection{Quantum Random-Access Memory (QRAM)}\label{subsec-qram}
In this subsection, we consider a simple QRAM \cite{Jaq23} that performs the following transformation: for any data set $d[0:N]$ with $N=2^n -1$, and for any address $0\leq j \leq N$, 
\begin{equation}
    |j\rangle |d[0:N]\rangle \mapsto |j\rangle |d[j]\rangle |d[0:j-1]\rangle |d[0]\rangle |d[j+1:N]\rangle.
    \label{eq:QRAM}
\end{equation}
Intuitively, given an address $j$, the desired data element $d[j]$ is swapped out. Using our programming language defined in Section \ref{sec-language}, a 
simple (but not very efficient) implementation of QRAM can be written as the following quantum recursive program.
\begin{exam}\label{exam-qram} We use $q_a[1:n]$ for the address register holding $|j\rangle$ 
and $q_d[0:N]$ for the data register holding quantum state $|d[0:N]\rangle$ on the LHS of \eqref{eq:QRAM}. Let
\begin{equation}
\label{eq:QRAM-prog}
\begin{split}
    U(l,r,k)\Leftarrow \ &\mathbf{if}\ k\leq n \ \mathbf{then}\\ 
    &\ \ \ \mathbf{begin\ local}\ m:=\lfloor(l+r)/2\rfloor;\\
    &\ \ \ \ \ \ \mathbf{qif}[q_a[k]]\ |0\rangle\rightarrow U(l,m,k+1)\\
    &\ \ \ \ \ \qquad\quad\ \square\ \ \ |1\rangle\rightarrow U(m+1,r,k+1);\\
    &\ \ \ \ \quad\quad\quad\quad\ \ \ \mathit{SWAP}[q_d[l],q_d[m+1]];\\
    &\ \ \ \ \quad\quad\quad\quad\ \ \ U(m+1,r,k+1)\\
    &\ \ \ \ \ \ \mathbf{fiq}\\
    &\ \ \ \mathbf{end}\\ 
    &\mathbf{fi}
\end{split}\end{equation}
Then one calls $U(0,N,1)$ for the QRAM operation.

It is particularly interesting to note that a Divide-and-Conquer strategy
was employed in this example where QRAM is divided
into two subproblems smaller than the original one that are
then solved respectively in each of the branches of a quantum
case statement. 

We define the following parameterized quantum states:
\begin{align*}
    &|D\rangle : n \mapsto |d[0:N]\rangle\\
    &|D'\rangle : (n,l,r)\mapsto |d[0:l-1]\rangle |d[r]\rangle |d[l+1:r-1]\rangle |d[l]\rangle |d[r+1:N]\rangle.
\end{align*}
Let $j$ be a classical array of type $\mathbf{integer}\rightarrow \mathbf{Boolean}$.
Then the correctness of the QRAM program in (\ref{eq:QRAM-prog}) can be expressed as the Hoare triple:
\begin{equation}
    \label{eq:QRAM-cor}
        k\leq n \wedge r-l+1=2^{n-k}\models_{\mathit{tot}}
        \left\{|j[k:n]\rangle |D\rangle (n)\right\}U(l,r,k)\left\{|j[k:n]\rangle |D'\rangle(l,l+j[k:n])\right\},
\end{equation}
where 
$|j[k:n]\rangle$ is the quantum state of the quantum array section $q_a[k:n]$,
and as usual we also use $j[k:n]$ to denote $\sum_{t=k}^n j[t]\cdot 2^{n-k}$, slightly abusing the notation.
\end{exam}

In the following we verify (\ref{eq:QRAM-cor}) using the proof rules in Section \ref{proof-system}. 
To apply the rule (Recursion-Par-S),
we define the ranking function $s=n-k$ and introduce the assumption that
\begin{equation}
    \label{eq:QRAM-ass}
    \begin{split}
        \vdash_{\mathit{tot}}&
        \left\{k\leq n \wedge r-l+1=2^{s}\wedge s<z,|j[k:n]\rangle |D\rangle(n)\right\}U(l,r,k)\\
        &\qquad\qquad \left\{k\leq n \wedge r-l+1=2^{s}, |j[k:n]\rangle |D'\rangle(l,l+j[k:n])\right\},
    \end{split}
\end{equation}
where $z$ is a fresh variable.
 
Since $k\leq n \rightarrow s\geq 0$ is obviously true,
with the rule (Block) and certain substitution of variables, it suffices to prove that
\begin{equation}
    \label{eq:QRAM-rec-goal}
    \begin{split}
        \vdash_{\mathit{tot}}&
        \left\{k\leq n \wedge r-l+1=2^{s}\wedge s=z,|j[k:n]\rangle |D\rangle(n)\right\}\mathbf{if}\ldots\mathbf{fi}\\
        &\qquad\qquad \left\{k\leq n \wedge r-l+1=2^{s}, |j[k:n]\rangle |D'\rangle(l,l+j[k:n])\right\},
    \end{split}
\end{equation}
where $\mathbf{if}\ldots\mathbf{fi}$ stands for the right hand side of the program in (\ref{eq:QRAM-prog}).
 
\textbf{Step I}:
First note that $\mathit{free}(s<z)\cap \mathit{change}\left(U(l,r,k)\right)=\emptyset$,
so using the rule (Invariance-Con) we can strengthen (\ref{eq:QRAM-ass}) to
\begin{equation}
    \label{eq:QRAM-ass-s}
    \begin{split}
        \vdash_{\mathit{tot}}&
        \left\{k\leq n \wedge r-l+1=2^{s}\wedge s<z,|j[k:n]\rangle |D\rangle(n)\right\}\ U(l,r,k)\\
        &\qquad\qquad \left\{k\leq n \wedge r-l+1=2^{s}\wedge s<z, |j[k:n]\rangle |D'\rangle(l,l+j[k:n])\right\}.
    \end{split}
\end{equation}
By applying the rule (Instantiation) to (\ref{eq:QRAM-ass-s}),
we have:
\begin{align}
    \label{eq:QRAM-rec-left}
    \begin{split}
        \vdash_{\mathit{tot}}&\left\{k+1\leq n \wedge m-l+1=2^{s-1}\wedge s-1<z,|j[k+1:n]\rangle |D\rangle (n)\right\} U(l,m,k+1)\\
            &\quad\left\{k+1\leq n \wedge m-l+1=2^{s-1}\wedge s-1<z,|j[k+1:n]\rangle|D'\rangle(l,l+j[k+1:n])\right\},
    \end{split}\\
    \label{eq:QRAM-rec-right}
    \begin{split}
        \vdash_{\mathit{tot}}&\left\{k+1\leq n \wedge r-m=2^{s-1}\wedge s-1<z,|j[k+1:n]\rangle |D\rangle (n)\right\}U(m+1,r,k+1)\\
            &\quad\left\{k+1\leq n \wedge r-m=2^{s-1}\wedge s-1<z,|j[k+1:n]\rangle |D'\rangle(m+1,m+1+j[k+1:n])\right\}.
    \end{split}
\end{align}

\textbf{Step II}: Using the rule (GA-P), we can prove:
\begin{equation}
    \label{eq:QRAM-swap}
    \begin{split}
        \vdash_{\mathit{tot}}&\left\{k+1\leq n \wedge r-m=2^{s-1}\wedge s-1<z,|j[k+1:n]\rangle |D\rangle(n) \right\}\mathit{SWAP}[q_d[l],q_d[m+1]]\\
        &\qquad \qquad\left\{k+1\leq n \wedge r-m=2^{s-1}\wedge s-1<z,|j[k+1:n]\rangle |D'\rangle(l,m+1)\right\}.
    \end{split}
\end{equation}
Note that 
\begin{align*}
    s=z&\models s-1<z,\\
    r-m=2^{s-1}&\models  m=\lfloor(l+r)/2\rfloor\wedge r-l+1=2^{s},\\
    m=\lfloor(l+r)/2\rfloor\wedge r-l+1=2^{s}&\models r-m=2^{s-1}.
\end{align*}
By the rules (SC-P) and (Consequence), we can combine (\ref{eq:QRAM-rec-right}) and (\ref{eq:QRAM-swap}) to obtain:
\begin{equation}
    \label{eq:QRAM-qif-right}
    \begin{split}
        \vdash_{\mathit{tot}}&\left\{k+1\leq n \wedge m=\lfloor(l+r)/2\rfloor\wedge r-l+1=2^{s}\wedge s=z,|j[k+1:n]\rangle |D\rangle(n) \right\}\\
            &\qquad U(m+1,r,k+1);\ \mathit{SWAP}[q_d[l],q_d[m+1]];\ U(m+1,r,k+1)\\
            &\qquad\qquad \left\{k+1\leq n \wedge  r-l+1=2^{s},|j[k+1:n]\rangle |D'\rangle(l,m+1+j[k+1:n])\right\}.
    \end{split}
\end{equation}
Similarly, note that
$m-l+1=2^{s-1}\models m=\lfloor(l+r)/2\rfloor\wedge r-l+1=2^{s}$ and $m=\lfloor(l+r)/2\rfloor\wedge r-l+1=2^{s}\models m-l+1=2^{s-1}$.
Now use the rules (QIF-P) and (Consequence) on (\ref{eq:QRAM-rec-left}) and (\ref{eq:QRAM-qif-right}), we have:
\begin{equation}
    \label{eq:QRAM-qif-fiq}
    \begin{split}
        \vdash_{\mathit{tot}}&\left\{k+1\leq n \wedge m=\lfloor(l+r)/2\rfloor\wedge r-l+1=2^{s}\wedge s=z,|j[k:n]\rangle |D\rangle (n)\right\}\\
            &\mathbf{qif}\ldots \mathbf{fiq}\left\{k+1\leq n \wedge r-l+1=2^{s},|j[k:n]\rangle |D'\rangle(l,l+j[k:n])\right\},
    \end{split}
\end{equation}
where $\mathbf{qif}\ldots\mathbf{fiq}$ stands for the corresponding part in the program in (\ref{eq:QRAM-prog}).
Note that by the rule (AS-P), we have:
\begin{equation}
    \label{eq:QRAM-m-as}
    \begin{split}
        \vdash_{\mathit{tot}}&\left\{k+1\leq n \wedge r-l+1=2^{s}\wedge s=z,|j[k:n]\rangle |D\rangle (n)\right\} m:=\lfloor(l+r)/2\rfloor\\
        &\qquad\left\{k+1\leq n\wedge m=\lfloor(l+r)/2\rfloor \wedge r-l+1=2^{s}\wedge s=z,|j[k:n]\rangle |D\rangle(n)\right\}.
    \end{split}
\end{equation}
Applying the rule (SC-P) to combine (\ref{eq:QRAM-m-as}) and (\ref{eq:QRAM-qif-fiq}) yields:
\begin{equation}
    \begin{split}
        \vdash_{\mathit{tot}}&\left\{k+1\leq n \wedge r-l+1=2^{s}\wedge s=z,|j[k:n]\rangle |D\rangle(n) \right\} m:=\lfloor(l+r)/2\rfloor;\ \mathbf{qif}\ldots \mathbf{fiq}\\
            &\qquad \qquad \left\{k+1\leq n \wedge r-l+1=2^{s},|j[k:n]\rangle |D'\rangle(l,l+j[k:n])\right\},
    \end{split}
\end{equation}
which by the rule (Block) further implies:
\begin{equation}
    \label{eq:QRAM-b-e-k+1}
    \begin{split}
        \vdash_{\mathit{tot}}&\left\{k+1\leq n \wedge r-l+1=2^{s}\wedge s=z,|j[k:n]\rangle |D\rangle(n) \right\}\\
            &\qquad \mathbf{begin}\ \mathbf{local}\ m:=\lfloor(l+r)/2\rfloor;\ \mathbf{qif}\ldots \mathbf{fiq}\ \mathbf{end}\\
            &\qquad \qquad \left\{k+1\leq n \wedge r-l+1=2^{s}, |j[k:n]\rangle |D'\rangle(l,l+j[k:n])\right\}.
    \end{split}
\end{equation}

\textbf{Step III}:
Now consider the base case. By the rules (SK-P)
and (IF-P), we have
\begin{align}
        \vdash_{\mathit{tot}}&
        \left\{k=n\wedge l=m\wedge s=z, |D\rangle(n)\right\}U(l,m,k+1)\left\{k=n\wedge l=m\wedge s=z, |D\rangle(n)\right\},\\
        \vdash_{\mathit{tot}}&
        \left\{k=n\wedge r=m+1\wedge s=z, |D\rangle(n)\right\}U(m+1,r,k+1)\left\{k=n\wedge r=m+1\wedge s=z, |D\rangle(n)\right\}.
\end{align}
Following similar steps to the above, we can derive:
\begin{equation}
    \label{eq:QRAM-b-e-n}
    \begin{split}
        \vdash_{\mathit{tot}}& \left\{k=n\wedge r-l+1=2\wedge s=z,|j[k:n]\rangle |D\rangle (n)\right\}\\
            &\qquad \mathbf{begin}\ \mathbf{local}\ m:=\lfloor(l+r)/2\rfloor;\ \mathbf{qif}\ldots \mathbf{fiq}\ \mathbf{end}\\
            &\qquad \qquad \left\{k=n\wedge r-l+1=2, |j[k:n]\rangle |D'\rangle(l,l+j[k:n])\right\},
    \end{split}
\end{equation}
which combined with (\ref{eq:QRAM-b-e-k+1}), by the rules (Consequence) and (Disjunction), gives
\begin{equation}
    \label{eq:QRAM-b-e-all}
    \begin{split}
        \vdash_{\mathit{tot}}&\left\{k\leq n \wedge r-l+1=2^{s}\wedge s=z,|j[k:n]\rangle |D\rangle(n) \right\}\\
            &\qquad \mathbf{begin}\ \mathbf{local}\ m:=\lfloor(l+r)/2\rfloor;\ \mathbf{qif}\ldots \mathbf{fiq}\ \mathbf{end}\\
            &\qquad \qquad \left\{k\leq n \wedge r-l+1=2^{s}, |j[k:n]\rangle |D'\rangle(l,l+j[k:n])\right\}.
    \end{split}
\end{equation}

Finally, using the rules (IF-P) and (SK-P) on (\ref{eq:QRAM-b-e-all}) immediately yields (\ref{eq:QRAM-rec-goal}).

\section{Discussions about Quantum Hoare Triples}

\subsection{Formal Quantum States}\label{formal-states}

Parameterized quantum states used in preconditions and postconditions in the Hoare triples considered in this paper are defined in Subsection \ref{sec-parameter} as purely semantic entities. They are easy to use in human reasoning like what we did in the  examples given in Section \ref{sec-examples} and Appendix \ref{appendix-examples}. However, in order to formalize our proof system in a proof assistant (e.g. Coq, Isabelle/HOL or Lean), it is better to define a syntax of parameterized quantum states.   
 In this appendix, we introduce formal quantum states as the syntax of parameterized quantum states. 
 
First, we slightly expand our quantum programming language $\mathbf{RQC}^{++}$. In Subsection \ref{sec-circuits}, $\mathcal{U}$ was defined as a set of non-parameterized elementary unitary gates. Now we assume $\mathcal{U}$ to be a set of basic parameterised quantum gate symbols. Each gate symbol $U\in\mathcal{U}$ is equipped with: \begin{enumerate}\item[(i)] a classical type of the form $T_1\times ...\times T_m$, called the parameter type, where $m\geq 0$; and \item[(ii)] a quantum type of the form $\hs_1\otimes ...\otimes\hs_n$, where $n\geq 1$ is called the arity of $U$.\end{enumerate} 
In the case of $m=0$, $U$ is a non-parameterised gate symbol.  
If for each $k\leq m$, $t_k$ is a classical term of type $T_k$, and for each $i\leq n$, $q_i$ is a simple or subscripted quantum variable with type $\hs_i$, then $U(t_1,...,t_m)[q_1,...,q_n]$ is called a quantum gate. Intuitively, it means that the quantum gate $U$ with parameters $t_1,...,t_m$ acts on quantum variables $q_1,...,q_n$. Then we have: 
 
 \begin{defn}Formal quantum states are inductively defined as follows:\begin{enumerate}\item If $t$ is a classical term of type $T$, and $q$ is a simple or subscripted quantum variable of type $\hs_T={\rm span} \{|v\rangle:v\in T\}$ (the Hilbert space with $\{|v\rangle:v\in T\}$ as its orthonormal basis), then $s=|t\rangle_q$ is a formal quantum state, and its signature is $\mathit{Sign}(s)=\{q\}$;
 
 \item If $s_1,s_2$ are formal quantum states and $\mathit{Sign}(s_1)\cap \mathit{Sign}(s_2)=\emptyset$, then $s_1\otimes s_2$ (oftenn simply written $s_1s_2$) is a formal quantum state too, and its signature is $\mathit{Sign}(s_1\otimes s_2)=\mathit{Sign}(s_1)\cup\mathit{Sign}(s_2)$;
 
 \item If $s_1, s_2$ are formal quantum states such that $\mathit{Sign}(s_1)=\mathit{Sign}(s_2)$, and $\alpha_1,\alpha_2$ are two classical terms of type $\mathbb{C}$ (complex numbers), then $\alpha_1s_1+\alpha_2 s_2$ is a formal quantum states, and its signature is $\mathit{Sign}(\alpha_1s_1+\alpha_2s_2)=\mathit{Sign}(s_1)=\mathit{Sign}(s_2)$; and 
 
 \item If $s$ is a formal quantum state and $U(\overline{t})[\overline{q}]$ is a quantum gate defined above with $\overline{q}\subseteq\mathit{Sign}(s)$, then $(U(\overline{t})[\overline{q}])s$ is a formal quantum states, and its signature is $\mathit{Sign}((U(\overline{t})[\overline{q}])s)=\mathit{Sign}(s).$
 \end{enumerate}
 \end{defn}
 
 The semantics of formal quantum states are then actually parameterized quantum states.
 
 \begin{defn}Let $s$ be a formal quantum state. Given an interpretation $\mathbb{I}$ of the assertion language. For any classical state $\sigma$ in $\mathbb{I}$, if $\sigma\models \mathit{Dist}(\mathit{Sign}(s))$ (see equation (\ref{def-dist}) for its definition), then the semantics $\sigma(s)$ of $s$ in state $\sigma$ is a vector in Hilbert space $\hs_{\mathit{Sign}(s)}=\bigotimes_{q\in\sigma(\mathit{Sign}(s))}\hs_q$, where $\sigma(\mathit{Sign}(s))=\{\sigma(q^\prime)|q^\prime\in\mathit{Sign}(s)\}$, and $\sigma(q^\prime)$ is defined as in Subsection \ref{sec-sem-q++}, and it is inductively defined as follows:   
 \begin{enumerate}\item $\sigma(|t\rangle_q)=|\sigma(t)\rangle_{\sigma(q)}\in\hs_{\sigma(q)}$; 
 \item $\sigma(s_1\otimes s_2)=\sigma(s_1)\otimes \sigma(s_2)$;
 \item $\sigma(\alpha_1s_1+\alpha_2s_2)=\sigma(\alpha_1)\sigma(s_1)+\sigma(\alpha_2)\sigma(s_2)$; and
 \item if $\sigma\models\mathit{Dist}(\overline{q})$, then $\sigma\left((U(\overline{t})[\overline{q}])s\right)=\left(U(\sigma(\overline{t}))[\sigma(\overline{q})]\right)\sigma(s),$ where $\sigma(\overline{t})=\sigma(t_1),...,\sigma(t_m)$ if $\overline{t}=t_1,...,t_m$, and $\sigma(\overline{q})=\sigma(q_1),...,\sigma(q_n)$ if $\overline{q}=q_1,...,q_n$. 
\end{enumerate}
 \end{defn}
 
It should be noticed that $\sigma(s)$ is not always well-defined because the distinction condition in clause (4) may not be satisfied. Whenever it is well-defined, it is a vector in the Hilbert space $\hs_{\mathit{Sign}(s)}$, but it may not be a quantum state. If $\|\sigma(s)\|=1$, then we say that $\sigma(s)$ is legitimate. In this case, $\sigma(s)$ is a pure quantum state. Therefore, the semantics $\llbracket s\rrbracket_\mathbb{I}$ of $s$ in the interpretation $\mathbb{I}$ can be understood as the parameterized quantum state:
 $$\llbracket s\rrbracket_\mathbb{I}:\sigma\mapsto\llbracket s\rrbracket_\mathbb{I}(\sigma)=\sigma(s).$$
 
 \textit{In particular, using formal quantum states introduced in this section, the conditions $\mathit{Inv}(|\psi\rangle:\overline{x})$ in the proof rules (Block) and (Substitution), the condition $\mathit{Inv}(|\psi\rangle:x)$ in the rule ($\exists$-introduction), the condition $\mathit{qv}(|\psi\rangle)\cap\mathit{qv}(C)=\emptyset$ in the rule (Invariance), and the substitutions $|\varphi\rangle[\overline{x}:=\overline{t}]$, $|\psi\rangle[\overline{x}:=\overline{t}]$ in the rule (Instantiation) can be defined syntactically rather than semantically. Here, we omit these straightforward definitions.}
 
\subsection{Hoare Triples with Quantum Predicates}

In this paper, we only consider Hoare triple of the form $\{A,|\varphi\rangle\}\ C\ \{B,|\psi\rangle$, where $|\varphi\rangle, |\psi\rangle$ are parameterized quantum states. Indeed, we can define a more general form of Hoare triple for quantum circuits parameterized by classical variables:
\begin{equation}\label{Hoare-BvN}\{A,X\}\ C\ \{B, Y\}\end{equation} where $A, B$ are still first order logical formulas as in (\ref{Hoare-0+}), but $X,Y$ are are two (parameterized) closed subspaces of the Hilbert space of quantum circuit $C$ (or equivalently the projection operators onto the subspaces). In this case, $X,Y$ are actually two propositions in the sense of Birkhoff-von Neumann quantum logic \cite{BvN36}. It is easy to see that whenever $X,Y$ are both one-dimensional, say $X=|\varphi\rangle\langle\varphi|$ and $Y=|\psi\rangle\langle\psi|$, then Hoare triple (\ref{Hoare-BvN}) degenerates to (\ref{Hoare-0+}).   

Indeed, $X,Y$ can be even defined as two quantum predicates (i.e. observables modelled by positive Hermitian operators) \cite{DP06}. These more general preconditions and postconditions have been used in several version of quantum Hoare logic (see for example \cite{Ying11, Zhou19, Feng21}).   
In this paper, since only quantum circuits are dealt with and no measurements are involved, the restricted form (\ref{Hoare-0+}) of Hoare triples are strong enough for our purpose. However, whenever we add quantum measurements into our programming languages and we are concerned with the expectations of certain observables as in many quantum algorithms, then Hoare triples with quantum predicates will be needed. 

It is an interesting topic to extend the proof system to the case of Hoare triples with general quantum predicates. It seems that the idea of the above subsection can be easily generalized for symbolizing quantum predicates.   

\section{Proof of Soundness and Completeness Theorems}\label{proof-sound}

In this appendix, we give the proofs of the Soundness Theorems \ref{soundness} and the (Relative) Completeness Theorem \ref{completeness}.

\subsection{Proof by Lifting from a Classical Program Logic}
\begin{proof} (Outline) Our idea is to establish a link between the soundness and completeness of our proof system for recursive quantum circuits and those of the corresponding proof system for a classical programming language with recursion. 
Then we can lift the soundness and completeness of the logic for classical recursive programs to the case of recursive quantum circuits. 

As an example, let us only consider the completeness for partial correctness. The soundness and the completeness for total correctness can be proved in a similar way. Remember that in Section \ref{sec-sound-complete}, the assertion language $\mathcal{L}$ of the classical programming language used in this paper was expanded to a first-order logical language $\mathcal{L}^+$, for which the arithmetic interpretation $\mathbb{I}$ is given. We further expand $\mathcal{L}^+$ to $\mathcal{L}^{++}$ by adding:\begin{itemize}\item new (classical) variables $X,X^\prime,...$ for quantum states simply considered as complex vectors. More precisely, for each quantum registers $\overline{q}$ (a finite sequence of quantum variables), we introduce some classical variables, say $X_{\overline{q}},X^\prime_{\overline{q}},...$, of which the values are vectors in the Hilbert space $\hs_{\overline{q}}$ (seen as a classical domain); and 
\item new (classical) function symbols for basic quantum gates, which are simply considered as mappings from complex vectors to complex vectors.
For each quantum gate $U$ on quantum register $\overline{q}$, we introduce a classical function symbol $F_U$. Thus, for example, $F_U(X_{\overline{q}}), F_U(X^\prime_{\overline{q}}),...$ are terms in $\mathcal{L}^{++}.$ 
 \end{itemize}
We also assume that $\mathcal{L}^{++}$ is equipped with equality symbol \textquotedblleft $=$\textquotedblright\ between complex vectors. Then each quantum circuit in $\mathbf{QC}^+$ can be encoded as a term in $\mathcal{L}^{++}$ with both the classical variables and function symbols in $\mathcal{L}^+$ and the newly introduced variables for encoding quantum states and function symbols for encoding quantum gates. In this way, each recursively defined quantum circuit $C$ in $\mathbf{RQC}^{++}$ can be appropriately translated into a recursive classical program $S(C)$ defined based on the first-order language $\mathcal{L}^{++}$ such that for any interpretation $\mathbb{I}$ of $\mathcal{L}^+$, 
\begin{equation}\label{Hoare-eq+1}\mathbb{I}\models_\mathit{par}\{A,|\varphi\rangle\}\ C\ \{B,|\psi\rangle\}\ {\rm iff}\ \mathbb{I}^\ast\models_\mathit{par}\{A\wedge X=|\varphi\rangle\}\ S(C)\ \left\{B\wedge X^\prime=|\psi\rangle\right\},\end{equation} where variable $X^\prime$ is introduced to store the part of complex vectors in the computational result of $C$; that is, $\llbracket C\rrbracket_\sigma(|\varphi\rangle(\sigma))$ in Lemma \ref{lem-linear1}, for a given classical state $\sigma$ (note that $|\psi\rangle$ is a parameterized quantum state), and $\mathbb{I}^\ast$ is the interpretation of $\mathcal{L}^{++}$ obtained from $\mathbb{I}$ by adding the interpretation of the new function symbols and the valuation of new variables $X,X^\prime,...$. The same conclusion holds for total correctness. 
This claim can be proved by induction on the length of $C$, and we omit the tedious details here. 

Now assume that $\mathbb{I}$ is an arithmetic interpretation and $\mathbb{I}\models_\mathit{par}\{A,|\varphi\rangle\}\ C\ \{B,|\psi\rangle\}$. 
First, it follows from (\ref{Hoare-eq+1}) that \begin{equation}\label{Hoare-eq+2}\mathbb{I}^\ast\models_\mathit{par}\{A\wedge X=|\varphi\rangle\}\ S(C)\ \left\{B\wedge X^\prime=|\psi\rangle\right\}.\end{equation}
It is easy to make $\mathbb{I}^\ast$ being arithmetic by appropriately extending the encoding of the finite sequences of elements in the domain to include the new elements of complex vectors. Then by the completeness theorem of the logic for classical recursive programs (see for example \cite{Francez}, Section 6.3 and \cite{AB90}, Theorem 5.6), we obtain: 
\begin{equation}\label{Hoare-eq+3}\mathit{Th}(\mathbb{I}^\ast)\vdash_\mathit{par}\{A\wedge X=|\varphi\rangle\}\ S(C)\ \left\{B\wedge X^\prime=|\psi\rangle\right\}\end{equation} from (\ref{Hoare-eq+2}). By the translation between quantum circuits in $\mathbf{RQC}^{++}$ and recursive classical programs defined based on $\mathcal{L}^{++}$ described above, we are able to translate each axiom or inference rule used in the proof of (\ref{Hoare-eq+3}) to an equivalent axiom or rule in our proof system $\mathcal{PS}_\mathit{par}$ for $\mathbf{RQC}^{++}$. Accordingly, the proof of (\ref{Hoare-eq+3}) is translated to a proof of 
\begin{equation}\mathit{Th}(\mathbb{I})\vdash_\mathit{par}\{A, |\varphi\rangle\}\ C\ \left\{B,|\psi\rangle\right\}\end{equation} in our proof system $\mathcal{PS}_\mathit{par}$. We omit the tedious translation here. Thus, the completeness for partial correctness is proved. 
\end{proof}

It must be emphasised that the meaning of a quantum circuit $C$ in physics was lost in the encoding $S(C)$ of $C$ in a classical programming language used in the above proof, although the encoding establishes a close link (\ref{Hoare-eq+1}) between the logical correctness of $C$ and its translation $S(C)$ in the classical language. Indeed, quantum states and gates are simply treated as complex vectors and matrices in $S(C)$, and quantum computation in $C$ is treated as classical operations of matrices. 

{\vskip 4pt}

\textbf{Direct Proof of Soundness and Completeness}: Except for the above proof, the soundness and completeness can also be proved by a direct way:
\begin{itemize}\item \textbf{\textit{Soundness}}: To prove the soundness, we only need to verify the validity of each axiom and that each of our proof rules preserves the validity.  

\item \textbf{\textit{Completeness}}: We can develop a theory of (liberal) weakest preconditions with respect to correctness formulas of the form (\ref{Hoare-0+}) for recursive quantum programs. Then we can prove the completeness following the ideas of the proof of completeness for program logics of classical recursive programs (see for example \cite{AB90}, Section 5 or \cite{Francez}, Section 6.3).\end{itemize}

Here, we are not going to elaborate the proof of completeness using weakest preconditions. However, the direct proofs of the validity of our axioms and proof rules are instructive for both better understanding and proper applications of them (in particular, those rules for recursion). So, we present these proofs in the following subsections.  

{\vskip 3pt}

\textbf{Convention}: When proving the soundness of our proof rules, we always assume an interpretation $\mathbb{I}$ of the assertion language, and then talk about partial correctness $\mathbb{I}\models_\mathit{par}$ and total correctness $\mathbb{I}\models_\mathit{tot}$ under the interpretation $\mathbb{I}$ (see Definition \ref{def-hoare-semantics}). To simplify the presentation, the notation $\mathbb{I}\models_\mathit{par}$ and $\mathbb{I}\models_\mathit{tot}$ will often be written as $\models_\mathit{par}$ and $\models_\mathit{tot}$, respectively, provided that no confusion happens. 

\subsection{Proof of Soundness of Structural Rules}

In this section, we prove the soundness of the structural rules presented in Table \ref{circuit-rules-proof-0}. To this end, let us first give a technical lemma.

\begin{lem}\label{lem-linear}\begin{enumerate}\item If $(C,\sigma,|\varphi_1\rangle)\rightarrow (C^\prime,\sigma^\prime,|\varphi_1^\prime\rangle)$, then for any $|\varphi_2\rangle$, we have $(C,\sigma,|\varphi_2\rangle)\rightarrow (C^\prime,\sigma^\prime,|\varphi_2^\prime\rangle)$ for some $|\varphi_2^\prime\rangle$. \item If $(C,\sigma,|\varphi_1\rangle)\rightarrow (C^\prime,\sigma^\prime,|\varphi_i^\prime\rangle)$ $(i=1,2)$, then it holds that $$(C,\sigma,\alpha_1|\varphi_1\rangle+\alpha_2|\varphi_2\rangle)\rightarrow (C^\prime,\sigma^\prime,\alpha_1|\varphi_1^\prime\rangle+\alpha_2|\varphi_2\rangle).$$
\end{enumerate}\end{lem}
\begin{proof} The proof is routine by induction on the length of $C$. We omit the tedious details here. 
\end{proof}

Now we are ready to present: 

\begin{proof}[Proof of rule (Linearity)] Assume that \begin{equation}\label{lin-1}\models_\mathit{par}\{A,|\varphi_i\rangle\}\ C\ \{B,|\psi_i\rangle\}\ (i=1,2).\end{equation} We want to prove:
\begin{equation}\label{lin-2}\models_\mathit{par}\{A,\alpha_1|\varphi_1\rangle+\alpha_2|\varphi_2\rangle\}\ C\ \{B,\alpha_1|\psi_1\rangle+\alpha_2|\psi_2\rangle\}.\end{equation} 

For any $\sigma\models A$, if $$(C,\sigma,\alpha_1|\varphi_1\rangle(\sigma)+\alpha_2|\varphi_2\rangle(\sigma))\rightarrow^\ast(\downarrow,\sigma^\prime,|\psi^\prime\rangle),$$ then using Lemma \ref{lem-linear}(1), we can assert that for $i=1,2$, there exists some $|\psi_i^\prime\rangle$ such that 
$$(C,\sigma,|\varphi_i\rangle)\rightarrow^\ast(\downarrow,\sigma^\prime,|\psi_i^\prime\rangle).$$ By assumption (\ref{lin-1}), we obtain $\sigma^\prime\models B$ and $|\psi_i^\prime\rangle=|\psi_i\rangle(\sigma^\prime)$. Thus, by Lemma \ref{lem-linear}(2), we have: $$|\psi^\prime\rangle=\alpha_1|\psi_1\rangle(\sigma^\prime)+\alpha_2|\psi_2\rangle(\sigma^\prime).$$ Therefore, the conclusion (\ref{lin-2}) holds.\end{proof}

\begin{proof}[Proof of rule (Consequence)] We only prove the validity of the rule in the case of partial correctness. It can be similarly proved in the case of total correctness. We assume: \begin{align}\label{cons-1}&\{A,|\varphi\rangle\}\models\{A^\prime,|\varphi^\prime\rangle\},\\
\label{cons-2}&\models_\mathit{par}\{A^\prime,|\varphi^\prime\rangle\}\ C\ \{B^\prime, |\psi^\prime\rangle\},\\ \label{cons-3}&\{B^\prime,|\psi^\prime\rangle\}\models\{B,|\psi\rangle\}.\end{align} We want to show that $$\models_\mathit{par}\{A,|\varphi\rangle\}\ C\ \{B, |\psi\rangle\};$$ that is, for any $\sigma\models A$, if \begin{equation}\label{cons-4}(C,\sigma,|\varphi\rangle(\sigma))\rightarrow^\ast(\downarrow,\sigma^\prime,|\psi^{\prime\prime}\rangle),\end{equation} then $\sigma^\prime\models B$ and $|\psi^{\prime\prime}\rangle=|\psi\rangle(\sigma^\prime)$. 

First, (\ref{cons-1}) implies $A\models A^\prime$ and $A^\prime\models |\varphi\rangle=|\varphi^\prime\rangle$. Then $\sigma\models A^\prime$ and $|\varphi\rangle=|\varphi^\prime\rangle$ in state $\sigma$. Then it follows from (\ref{cons-4}) that \begin{equation}\label{cons-5}(C,\sigma,|\varphi^\prime\rangle(\sigma))\rightarrow^\ast(\downarrow,\sigma^\prime,|\psi^{\prime\prime}\rangle),\end{equation}
With (\ref{cons-2}) and (\ref{cons-5}), we obtain $\sigma^\prime\models B^\prime$ and $|\psi^{\prime\prime}\rangle=|\psi^\prime\rangle(\sigma^\prime)$. By (\ref{cons-3}) we have $B^\prime\models B$ and $B^\prime\models|\psi^\prime\rangle=|\psi\rangle$. Consequently, $\sigma^\prime\models B$ and $|\psi^{\prime\prime}\rangle=|\psi\rangle(\sigma^\prime)$. \end{proof}

\subsection{Proof of Soundness of Rules for Non-recursive Quantum Circuits}

In this section, we prove the soundness of the rules presented in Table \ref{circuit-rules-proof-1}. The soundness proofs of the rules (SK-P) and (GA-P) are routine, and the soundness proofs of the rules (AS-P), (SC-P) and (IF-P) are similar to the case of classical Hoare logic. So, they are omitted here, and we only present: 
\begin{proof}[Proof of Rule (QIF-P)] For simplicity, let $$\mathbf{qif} ... \mathbf{fiq}\equiv \mathbf{qif}[q] (|0\rangle\rightarrow C_0)\ \square\ (|1\rangle\rightarrow C_1)\ \mathbf{fiq}.$$ 
Assume that 
\begin{equation}\label{qif-1}\models_\mathit{par}\{A,|\varphi_i\rangle\}\ C_i\ \{B,|\psi_i\rangle\}\ (i=0,1).\end{equation} We want to prove \begin{equation}\label{qif-2}\models_\mathit{par}\left\{A,\alpha_0|0\rangle_q|\varphi_0\rangle+\alpha_1|1\rangle_q|\varphi_1\rangle\right\}\ \mathbf{qif}... \mathbf{fiq}\ \left\{B,\alpha_0|0\rangle_q|\psi_0\rangle+\alpha_1|1\rangle_q|\psi_1\rangle\right\}.\end{equation}

For any $\sigma\models A$, if 
$$(\mathbf{qif}...\mathbf{fiq},\sigma, \alpha_0|0\rangle_{\sigma(q)}|\varphi_0\rangle(\sigma)+\alpha_1|1\rangle_{\sigma(q)}|\varphi_1\rangle(\sigma))\rightarrow^\ast(\downarrow,\sigma^\prime,|\psi^\prime\rangle),$$ then we must have that for $i=0,1$, it holds that 
$$(C_i,\sigma,|\varphi_i\rangle(\sigma))\rightarrow^\ast(\downarrow,\sigma^\prime,|\psi_i^\prime\rangle)$$ for some $|\psi^\prime_i\rangle$, and $$|\psi^\prime\rangle=\alpha_0|0\rangle_{\sigma(q)}|\psi_0^\prime\rangle+\alpha_1|1\rangle_{\sigma(q)}|\psi_1^\prime\rangle.$$ By assumption (\ref{qif-1}), we obtain $\sigma^\prime\models B$ and $|\psi_i^\prime\rangle=|\psi_i\rangle(\sigma^\prime)$ for $i=0,1$. 
On the other hand, by the assumption that $\mathit{cv}(q)\cap\mathit{change}(C_i)=\emptyset$, we obtain $\sigma(q)=\sigma^\prime(q)$. Then $$|\psi^\prime\rangle=\alpha_0|0\rangle_{\sigma(q)}|\psi_0\rangle(\sigma^\prime)+\alpha_1|1\rangle_{\sigma(q)}|\psi_1\rangle(\sigma^\prime)=\alpha_0|0\rangle_{\sigma^\prime(q)}|\psi_0\rangle(\sigma^\prime)+\alpha_1|1\rangle_{\sigma^\prime(q)}|\psi_1\rangle(\sigma^\prime),$$ and the conclusion (\ref{qif-2}) holds.\end{proof}

\subsection{Proof of Soundness of Rules (Block) and (Instantiation)}

Let us first prove the soundness of the rule (Block): 

\begin{proof}[Proof of Rule (Block)] Assume that \begin{equation}\label{assume-block}\models_\mathit{par}\{A,|\varphi\rangle\}\ \overline{x}:=\overline{t};C\ \{B,|\psi\rangle\}. 
\end{equation} For any $\sigma\models A$, if 
$$(\mathbf{begin\ local}\ \overline{x}:\overline{t},C\ \mathbf{end},\sigma,|\varphi\rangle(\sigma))\rightarrow^\ast(\downarrow,\sigma^\prime,|\psi^\prime\rangle),$$ then for some $\sigma^{\prime\prime}$, it holds that \begin{equation}\label{block-1}(\overline{x}:=\overline{t};C,\sigma,|\varphi\rangle(\sigma))\rightarrow^\ast(\downarrow,\sigma^{\prime\prime},|\psi^\prime\rangle)\end{equation} and \begin{align*}
(\mathbf{begin\ local}\ \overline{x}:\overline{t},C\ \mathbf{end},\sigma,|\varphi\rangle(\sigma))&\rightarrow (\overline{x}:=\overline{t};C,\sigma,|\varphi\rangle(\sigma))\\
&\rightarrow^\ast (\overline{x}:=\sigma(\overline{x}),\sigma^{\prime\prime},|\psi^\prime\rangle)\\ 
&\rightarrow (\downarrow,\sigma^\prime,|\psi^\prime\rangle)
\end{align*} with $\sigma^\prime=\sigma^{\prime\prime}[\overline{x}:=\sigma(\overline{x})]$. From (\ref{assume-block}) and (\ref{block-1}), it follows that $\sigma^{\prime\prime}\models B$ and $|\psi^\prime\rangle=|\psi\rangle(\sigma^{\prime\prime})$. By the assumption $\overline{x}\cap\mathit{free}(B)=\overline{x}\cap\mathit{var}(|\psi\rangle)=\emptyset$, we obtain $\sigma^\prime=\sigma^{\prime\prime}[\overline{x}:=\sigma(\overline{x})]\models B$, and $|\psi^\prime\rangle=|\psi\rangle(\sigma^\prime)$. Therefore, we have: $$\models_\mathit{par}\{A,|\varphi\rangle\}\ \mathbf{begin\ local}\ \overline{x}:=\overline{t};C\ \mathbf{end}\ \{B,|\psi\rangle\}.$$ This proves that the rule (Block) is valid for partial correctness. Similarly, we can prove that the rule (Block) is also valid for total correctness. 
\end{proof}

To prove the soundness of the rule (Instantiation), we need the following technical lemma. 

\begin{lem}\label{proc-transition}Let procedure identifier $P$ be declared by $P(\overline{u})\Leftarrow C$ and $\overline{x}\cap\mathit{var}(C)=\mathit{var}(\overline{t})\cap\mathit{change}(C)=\emptyset$. If $$\left(P(\overline{t}),\sigma,|\varphi\rangle[\overline{x}:=\overline{t}]\right(\sigma))\rightarrow^\ast(\downarrow, \sigma^\prime,|\psi\rangle),$$ then $$\left(P(\overline{x}),\sigma[\overline{x}:=\sigma(\overline{t})],|\varphi\rangle(\sigma[\overline{x}:=\sigma(\overline{t})])\right)\rightarrow^\ast(\downarrow,\sigma^\prime[\overline{x}:=\sigma^\prime(\overline{t})],|\psi\rangle).$$ \end{lem}

\begin{proof} This lemma can be proved in the same way as the proof of Lemma 5.5 in \cite{Apt09}.\end{proof}

Now we are ready to present: 

\begin{proof}[Proof of Rule (Instantiation)] We assume that 
\begin{equation}\label{assume-inst}\models_\mathit{par}\{A,|\varphi\rangle\}\ P(\overline{x})\ \{B,|\psi\rangle\}\end{equation} and want to show that \begin{equation}\label{concl-inst}\models_\mathit{par}\left\{A[\overline{x}:=\overline{t}],|\varphi\rangle[\overline{x}:=\overline{t}]\right\}\ P(\overline{t})\ \left\{B[\overline{x}:=\overline{t}],|\psi\rangle[\overline{x}:=\overline{t}]\right\}.\end{equation}

For any $\sigma\models A[\overline{x}:=\overline{t}]$, the substitution lemma in first-order logic (see \cite{Apt09}, Lemma 2.5) yields: \begin{equation}\label{lem-substitute}\sigma[\overline{x}:=\sigma(\overline{t})]\models A.\end{equation} 
If $$(P(\overline{t}),\sigma,|\varphi\rangle[\overline{x}:=\overline{t}](\sigma))\rightarrow^\ast(\downarrow,\sigma^\prime,|\psi^\prime\rangle),$$ then it follows from Lemma \ref{proc-transition} that 
\begin{equation}\label{mid-inst}\left(P(\overline{x}),\sigma[\overline{x}:=\sigma(\overline{t})],|\varphi\rangle\right(\sigma[\overline{x}:=\sigma(\overline{t})]))\rightarrow^\ast\left(\downarrow, \sigma^\prime[\overline{x}:=\sigma^\prime(\overline{t})],|\psi^{\prime}\rangle\right).\end{equation} 
Thus, with the assumption (\ref{assume-inst}) together with (\ref{lem-substitute}) and (\ref{mid-inst}) we obtain: 

\vskip 3pt

\textbf{Claim 1}: $\sigma^\prime[\overline{x}:=\sigma^\prime(\overline{t})]\models B$. Using the substitution lemma in first-order logic once again, we assert that $\sigma^\prime\models B[\overline{x}:=\overline{t}].$

\vskip 3pt

\textbf{Claim 2}: $|\psi^\prime\rangle=|\psi\rangle(\sigma^\prime[\overline{x}:=\sigma^\prime(\overline{t}])=|\psi\rangle[\overline{x}:=\overline{t}](\sigma^\prime)$. Here, the second equality comes directly from the definition of $|\psi\rangle[\overline{x}:=\overline{t}]$ (see the last paragraph of Subsection \ref{sec-rule-para}). 

\vskip 3pt 

Finally, we complete the proof by combining the above two claims. 
\end{proof} 

\subsection{Proof of Soundness of Rules (Recursion-Par) and (Recursion-Tot)}

The basic idea of the soundness proofs of the rules (Recursion-Par) and (Recursion-Tot) is the same as that for classical programs given in Sections 4.3 and 5.3 of \cite{Apt09}. 

For any set $\mathcal{D}$ of declarations, we write $\rightarrow_\mathcal{D}$ for the transition relation defined with respect to $\mathcal{D}$. Accordingly, we write $\models_\mathit{par}^\mathcal{D}$ and $\models_\mathit{tot}^\mathcal{D}$ for partial correctness and total correctness, respectively, defined based on $\rightarrow_\mathcal{D}$. Furthermore, for any set $\Sigma$ of Hoare triples, by the judgment $$\Sigma\models_\mathcal{D}\{A,|\varphi\rangle\}\ C\ \{B,|\psi\rangle\}$$ we mean:   
\begin{itemize}\item if $\models_\mathcal{D}\{A^\prime,|\varphi^\prime\rangle\}\ C^\prime\ \{B^\prime,|\psi^\prime\rangle\}$ for all Hoare triples $\{A^\prime,|\varphi^\prime\rangle\}\ C^\prime\ \{B^\prime,|\psi^\prime\rangle\}\in\Sigma,$ then $\models_\mathcal{D}\{A,|\varphi\rangle\}\ C\ \{B,|\psi\rangle\}$.
\end{itemize}

Now we are given a set $$\mathcal{D}=\{P_1(\overline{u})\Leftarrow C_1,...,P(\overline{u})\Leftarrow C_n\}$$ of declarations. Then the entailment $\models$ is defined by
$$\Sigma\models\{A,|\varphi\rangle\}\ C\ \{B,|\psi\rangle\}\ {\rm iff}\ \Sigma\models_{\mathcal{D}^\prime}\{A,|\varphi\rangle\}\ C\ \{B,|\psi\rangle\}\ {\rm for\ all}\ \mathcal{D}^\prime\ {\rm with}\ \mathit{var}(\mathcal{D}^\prime)\subseteq\mathit{var}(\mathcal{D}).$$

First, we have:

\begin{lem}[Soundness of $\vdash$]\label{promise-sound} If $\Sigma\vdash\{A,|\varphi\rangle\}\ C\ \{B,|\psi\rangle\}$ then $\Sigma\models\{A,|\varphi\rangle\}\ C\ \{B,|\psi\rangle\}$.
\end{lem}
\begin{proof}By the definition of provability $\vdash$ (see equation (\ref{provability})), it suffices to verify that all of the rules in Tables \ref{circuit-rules-proof-0} and \ref{circuit-rules-proof-1} preserve $\models$. The details are routine and thus omitted here.  
\end{proof}

To describe the unfolding in the execution of a quantum recursive circuit, we introduced the notions of substitution and syntactic approximation: 

\begin{defn}\label{def-sub}Let $C,D_1,...,D_n\in\mathbf{RQC}^{++}$ be quantum circuits. Then the simultaneous substitution $C[D_1/P_1,...,D_n/P_n]$ is defined by induction on the length of $C$:\begin{enumerate}
\item $C\equiv\mathbf{skip}, \overline{x}:=\overline{t}$ or $U[\overline{q}]$. Then $$C[D_1/P_1,...,D_n/P_n]\equiv C.$$
\item $C\equiv C_1;C_2$. Then $$C[D_1/P_1,...,D_n/P_n]\equiv C_1[D_1/P_1,...,D_n/P_n];C_2[D_1/P_1,...,D_n/P_n].$$
\item $C\equiv\mathbf{if}\ b\ \mathbf{then}\ C_1\ \mathbf{else}\ C_2\ \mathbf{fi}$. Then $$C[D_1/P_1,...,D_n/P_n]\equiv\mathbf{if}\ b\ \mathbf{then}\ C_1[D_1/P_1,...,D_n/P_n]\ \mathbf{else}\ C_2[D_1/P_1,...,D_n/P_n]\ \mathbf{fi}.$$
\item $C\equiv\mathbf{qif}[q]\ (|0\rangle\rightarrow C_0)\square(|1\rangle\rightarrow C_1)\ \mathbf{fiq}$. If $q\notin\bigcup_{i=1}^n\mathit{qv}(D_i)$, then 
\begin{equation}\label{qif-sub}\begin{split}C[D_1/P_1,...,D_n/P_n]\equiv\mathbf{qif}&[q]\ (|0\rangle\rightarrow C_0[D_1/P_1,...,D_n/P_n])\\ &\ \square\ \ (|1\rangle\rightarrow C_1[D_1/P_1,...,D_n/P_n])\ \mathbf{fiq}.\end{split}\end{equation}
\item $C\equiv\mathbf{begin\ local}\ \overline{x}:=\overline{t};C^\prime\ \mathbf{end}$. Then $$C[D_1/P_1,...,D_n/P_n]=\mathbf{begin\ local}\ \overline{y}:=\overline{t};C^\prime[\overline{y}/\overline{x}][D_1/P_1,...,D_n/P_n]\ \mathbf{end}$$ where $$\overline{y}\cap\left(\bigcup_{i=1}^n\mathit{var}(D_i)\right)=\emptyset.$$
\item $C\equiv P_i(\overline{t})$. Then $$C[D_1/P_1,...,D_n/P_n]\equiv\mathbf{begin\ local}\ \overline{u_i}:=\overline{t};C_i\ \mathbf{end}.$$
\end{enumerate}\end{defn}

It is worth noting that the condition $q\notin\bigcup_{i=1}^n\mathit{qv}(D_i)$ is necessary to warrant that (\ref{qif-sub}) is well-defined. 

\begin{defn} Given a set $\mathcal{D}=\{P_1(\overline{u})\Leftarrow C_1,...,P(\overline{u})\Leftarrow C_n\}$ of declarations. For any well defined quantum circuit $C\in\mathbf{RQC}^{++}$, its $k$th syntactic approximation $C^k$ is defined by induction on integer $k\geq 0$:\begin{enumerate}\item $C^0\equiv\mathbf{abort}$ --- a program such that for any classical state $\sigma$ and quantum state $|\varphi\rangle$, $(\mathbf{abort},\sigma,|\varphi\rangle)\not\rightarrow$;  

\item $C^{k+1}\equiv C[C_1^k/P_1,...,C_n^k/P_n]$ for $k\geq 0$. 
\end{enumerate}\end{defn}

Since $C$ is well-defined (see Definition \ref{def-well}), if $P_i$ appears in a sub-circuit $$\mathbf{qif}[q]\ (|0\rangle\rightarrow C_0)\square(|1\rangle\rightarrow C_1)\ \mathbf{fiq}$$ of $C$, then $q\notin\mathit{qv}(C_i)$. It is easy to show that $\mathit{qv}(C_i^k)\subseteq\mathit{qv}(C_i)$ by induction on $k$. Therefore, $q\notin\mathit{qv}(C_i^k)$, and $C^{k+1}$ is well-defined (see Definition \ref{def-sub}(4)).

For each integer $k\geq 0$, we define declarations: $$\mathcal{D}^k=\left\{P_1(\overline{u_1})\Leftarrow C_1^k,...P_n(\overline{u_n})\Leftarrow C_n^k\right\}.$$ Then we have: 

\begin{lem}\label{recursive-transition}For any quantum circuit $C\in\mathbf{RQC}^{++}$: \begin{enumerate}
\item $(C,\sigma,|\varphi\rangle)\rightarrow_\mathcal{D}^\ast(\downarrow,\sigma^\prime,|\psi\rangle)$ iff for some $k\geq 0,$ $(C^k,\sigma,|\varphi\rangle)\rightarrow^\ast(\downarrow,\sigma^\prime,|\psi\rangle);$  
\item $(C,\sigma,|\varphi\rangle)\rightarrow^\ast_{\mathcal{D}^k}(\downarrow,\sigma^\prime,|\psi\rangle)$ iff $(C^{k+1},\sigma,|\varphi\rangle)\rightarrow^\ast(\downarrow,\sigma^\prime,|\psi\rangle).$
\end{enumerate}\end{lem}
\begin{proof} These equivalences can be proved by induction on the length of transitions. The details are similar to the case of classical recursive programs (see for example Section 6.3 of \cite{Francez}) and thus omitted here.\end{proof}

Now we are ready to present the soundness proof of the rule (Recursion-Par):

\begin{proof}[Proof of Rule (Recursion-Par)] Assume:
\begin{enumerate}\item[(A)] $\{A_i,|\varphi_i\rangle\}\ P_i(\overline{t_i})\ \{B_i,|\psi_i\rangle\} (i=1,...,n)\vdash \{A,|\varphi\rangle\}\ C\ \{B,|\psi\rangle\};$
\item[(B)] $\{A_i,|\varphi_i\rangle\}\ P_i(\overline{t_i})\ \{B_i,|\psi_i\rangle\} (i=1,...,n)\vdash \{A_j,|\varphi_j\rangle\}\ \mathbf{begin\ local}\ \overline{u_j}:=\overline{t_j};C_j\ \mathbf{end}\ \{B_j,|\psi_j\rangle\}\ {\rm for}\ j=1,...,n.$
\end{enumerate} Then with Lemma \ref{promise-sound} we have:
\begin{enumerate}\item[(A')] $\{A_i,|\varphi_i\rangle\}\ P_i(\overline{t_i})\ \{B_i,|\psi_i\rangle\} (i=1,...,n)\models \{A,|\varphi\rangle\}\ C\ \{B,|\psi\rangle\};$
\item[(B')] $\{A_i,|\varphi_i\rangle\}\ P_i(\overline{t_i})\ \{B_i,|\psi_i\rangle\} (i=1,...,n)\models \{A_j,|\varphi_j\rangle\}\ \mathbf{begin\ local}\ \overline{u_j}:=\overline{t_j};C_j\ \mathbf{end}\ \{B_j,|\psi_j\rangle\}\ {\rm for}\ j=1,...,n.$
\end{enumerate}
Therefore, it suffices to derive the following claim from (B'): for each $j=1,...,n$, 

{\vskip 3pt}

\textbf{Claim 1}: $\models_\mathcal{D}\{A_j,|\varphi_j\rangle\}\ P_j(\overline{t_j})\ \{B_j,|\psi_j\rangle\}.$

{\vskip 3pt}

By definition, Claim 1 is equivalent to:

{\vskip 3pt}

\textbf{Claim 2}: For any $\sigma\models A_j$, if \begin{equation}\label{claim-2}(P_j(\overline{t_j}),\sigma,|\psi_j\rangle(\sigma))\rightarrow_\mathcal{D}^\ast(\downarrow,\sigma^\prime,|\psi_j^\prime\rangle),\end{equation} then $\sigma^\prime\models B_j$ and $|\psi_j^\prime\rangle=|\psi_j\rangle(\sigma^\prime)$. 

{\vskip 3pt}

It follows from Lemma \ref{recursive-transition} (1) and (2) that 
\begin{align*}Eq. (\ref{claim-2}) &\Leftrightarrow {\rm for\ some}\ k\geq 0, \left(P_j(\overline{t_j})^k,\sigma,|\varphi_j\rangle(\sigma)\right)\rightarrow^\ast(\downarrow,\sigma^\prime,|\psi_j^\prime\rangle)\\
&\Leftrightarrow \left(P_j(\overline{t_j})^0,\sigma,|\varphi_j\rangle(\sigma)\right)\rightarrow^\ast(\downarrow,\sigma^\prime,|\psi_j^\prime\rangle)\ {\rm or}\\ &\qquad\qquad \qquad {\rm for\ some}\ k\geq 0, \left(P_j(\overline{t_j})^{k+1},\sigma,|\varphi_j\rangle(\sigma)\right)\rightarrow^\ast(\downarrow,\sigma^\prime,|\psi_j^\prime\rangle)\\
&\Leftrightarrow {\rm for\ some}\ k\geq 0, (P_j(\overline{t_j}),\sigma,|\varphi_j\rangle(\sigma))\rightarrow_{\mathcal{D}^k}^\ast(\downarrow,\sigma^\prime,|\psi_j^\prime\rangle).
\end{align*}
Therefore, Claim 2 is equivalent to 

{\vskip 3pt}

\textbf{Claim 3}: For any $k\geq 0$, 
\begin{equation}\label{claim3}\models_{\mathcal{D}^k}\{A_j,|\varphi_j\rangle\}\ P_j(\overline{t_j})\ \{B_j,|\psi\rangle);\end{equation}
that is, if $\sigma\models A_j$ and 
$$(P_j(\overline{t_j}),\sigma,|\varphi_j\rangle(\sigma))\rightarrow_{\mathcal{D}^k}^\ast(\downarrow,\sigma^\prime,|\psi_j^\prime\rangle),$$ then $\sigma^\prime\models B_j$ and $|\psi_j^\prime\rangle=|\psi_j\rangle(\sigma^\prime)$. 

{\vskip 3pt}

Now we prove Claim 3 by induction on $k$. 

(1) The basis case of $k=0$ is trivial by the definition of program $\mathbf{abort}$.  

(2) For the induction step, let us assume that Claim 3 is true for $k$. We want to show that 
\begin{equation}\label{claim3-}\models_{\mathcal{D}^{k+1}}\{A_j,|\varphi_j\rangle\}\ P_j(\overline{t_j})\ \{B_j,|\psi\rangle);\end{equation}
that is, if $\sigma\models A_j$ and 
\begin{equation}\label{claim3--}(P_j(\overline{t_j}),\sigma,|\varphi_j\rangle(\sigma))\rightarrow_{\mathcal{D}^{k+1}}^\ast(\downarrow,\sigma^\prime,|\psi_j^\prime\rangle),\end{equation} then $\sigma^\prime\models B_j$ and $|\psi_j^\prime\rangle=|\psi_j\rangle(\sigma^\prime)$. 

First, note that $\mathit{var}\left(\mathcal{D}^k\right)\subseteq\mathit{var}(\mathcal{D})$. Then (\ref{claim3}) and assumption (B') imply:
\begin{equation}\label{claim3+}\models_{\mathcal{D}^k}\{A_j,|\varphi_j\rangle\}\ \mathbf{begin\ local}\ \overline{u_j}:=\overline{t_j};C_j\ \mathbf{end}\ \{B_j,|\psi_j\rangle\}.\end{equation} 

Second, it is easy to verify by induction on $k$ that 
\begin{equation}\label{claim3++}(\mathbf{begin\ local}\ \overline{u_j}:=\overline{t_j};C_j\ \mathbf{end})^{k+1}\equiv \mathbf{begin\ local}\ \overline{u_j}:=\overline{t_j};C_j^{k+1}\ \mathbf{end}.\end{equation}
Since $P_j(\overline{t_j})\Leftarrow C_j^{k+1}\in\mathcal{D}^{k+1}$, it follows from the transition rule (RC) that 
\begin{equation}\label{claim3+1}(P_j(\overline{t_j}),\sigma,|\varphi_j\rangle(\sigma))\rightarrow_{\mathcal{D}^{k+1}}(\mathbf{begin\ local}\ \overline{u_j}:=\overline{t_j};C_j^{k+1}\ \mathbf{end},\sigma,|\varphi_j\rangle(\sigma)).\end{equation}
Thus, we have: \begin{align*}Eq. (\ref{claim3--})&\Leftrightarrow \left(\mathbf{begin\ local}\ \overline{u_j}:=\overline{t_j};C_j^{k+1}\ \mathbf{end},\sigma,|\varphi_j\rangle(\sigma)\right)\rightarrow_{\mathcal{D}^{k+1}}^\ast(\downarrow,\sigma^\prime,|\psi_j^\prime\rangle)\\
&\Leftrightarrow \left((\mathbf{begin\ local}\ \overline{u_j}:=\overline{t_j};C_j\ \mathbf{end})^{k+1},\sigma,|\varphi_j\rangle(\sigma)\right)\rightarrow_{\mathcal{D}^{k+1}}^\ast(\downarrow,\sigma^\prime,|\psi_j^\prime\rangle)\qquad ({\rm By\ Eq. (\ref{claim3++})})\\ 
&\Leftrightarrow \left((\mathbf{begin\ local}\ \overline{u_j}:=\overline{t_j};C_j\ \mathbf{end})^{k+1},\sigma,|\varphi_j\rangle(\sigma)\right)\rightarrow^\ast(\downarrow,\sigma^\prime,|\psi_j^\prime\rangle)\\ 
&\Leftrightarrow (\mathbf{begin\ local}\ \overline{u_j}:=\overline{t_j};C_j\ \mathbf{end},\sigma,|\varphi_j\rangle(\sigma))\rightarrow_{\mathcal{D}^k}^\ast(\downarrow,\sigma^\prime,|\psi_j^\prime\rangle)\qquad ({\rm By\ Lemma\ \ref{recursive-transition}(1)})
\end{align*}
Consequently, if $\sigma\models A_j$ and (\ref{claim3--}) holds, then from (\ref{claim3+}) we obtain that $\sigma^\prime\models B_j$ and $|\psi_j^\prime\rangle=|\psi_j\rangle(\sigma^\prime)$.   
Therefore, (\ref{claim3-}) is proved, and we complete the proof. 
\end{proof}

Finally, we are able to prove the soundness of the rule (Recursion-Tot).

\begin{proof}[Proof of Rule (Recursion-Tot)] With an argument similar to that in the soundness proof of the rule (Recursion-Par), it suffices to derive
\begin{equation}\label{sound-rec-total}\models_\mathit{tot}\left\{A_j,|\varphi_j\rangle\right\}\ P_j(\overline{t_j})\ \left\{B_j,|\psi_j\rangle\right\}\ {\rm for}\ j=1,...,n\end{equation} from the following two assumptions: 
\begin{equation}\label{sound-rec-total1}\begin{split}&\left\{A_i\wedge r<z, |\varphi_i\rangle\right\}\ P_i(\overline{t_i})\ \left\{B_i,|\psi_i\rangle\right\}\ (i=1,...,n)\models_\mathit{tot}\\ &\qquad\qquad\left\{A_j\wedge r=z,|\varphi_j\rangle\right\}\ \mathbf{begin\ local}\ \overline{u_j}:=\overline{t_j};C_j\ \mathbf{end}\ \left\{B_j,|\psi_j\rangle\right\}\ {\rm for}\ j=1,...,n;\end{split}\end{equation}
\begin{equation}\label{sound-rec-total2}\models A_i\rightarrow r\geq 0.
\end{equation} Since $z$ is a fresh integer variable, with the assumption (\ref{sound-rec-total2}), what we need to prove is actually that for all integers $m\geq 0$,
\begin{equation}\label{sound-rec-total3}\models_\mathit{tot}\left\{A_j\wedge r<m,|\varphi_j\rangle\right\}\ P_j(\overline{t_j})\ \left\{B_j,|\psi_j\rangle\right\}\ {\rm for}\ j=1,...,n.\end{equation}

Let us prove (\ref{sound-rec-total3}) by induction on $m$. The basis case of $m=0$ is trivial. Now we assume that (\ref{sound-rec-total3}) is true for $m=k$ and want to show that it is also true for $m=k+1$. By the induction hypothesis for $m=k$ and the assumption (\ref{sound-rec-total1}), it holds that    
\begin{equation}\label{sound-rec-total4}\left\{A_j\wedge r=k,|\varphi_j\rangle\right\}\ \mathbf{begin\ local}\ \overline{u_j}:=\overline{t_j};C_j\ \mathbf{end}\ \left\{B_j,|\psi_j\rangle\right\}.\end{equation}
Then in a way similar to the derivation of Claim 1 in the soundness proof of the rule (Recursion-Par), we can derive:
\begin{equation}\label{sound-rec-total5}\models_\mathit{tot}\left\{A_j\wedge r=k,|\varphi_j\rangle\right\}\ P_j(\overline{t_j})\ \left\{B_j,|\psi_j\rangle\right\}\end{equation}
from (\ref{sound-rec-total4}). Putting (\ref{sound-rec-total5}) and the induction hypothesis of $m=k$ together, we see that (\ref{sound-rec-total3}) is true for $m=k+1$, and we complete the proof. 
\end{proof}

\subsection{Proofs of Soundness of Auxiliary Rules}\label{proof-auxiliary}

We choose to prove the soundness of some of the auxiliary rules in Table \ref{circuit-rules-proof-5}. The soundness proofs of other auxiliary rules are omitted because they are either trivial or similar to some one given below.  

\begin{proof}[Proof of Rule (Invariance)] We prove that the rule (Invariance) is sound for partial correctness. For any $\sigma\models A$, assume that $$(C,\sigma,|\varphi\rangle(\sigma))\rightarrow^\ast(\downarrow,\sigma^\prime,|\varphi^\prime).$$ By the assumption $\mathit{free}(A)\cap \mathit{change}(C)=\emptyset$, we have $\sigma|\mathit{free}(A)=\sigma^\prime|\mathit{free}(A)$. Then $\sigma\models A$ implies $\sigma^\prime\models A$. On the other hand, by the assumption $\mathit{qv}(|\varphi\rangle)\cap\mathit{qv}(C)=\emptyset$, we obtain $|\varphi^\prime\rangle=|\varphi\rangle$. Therefore, it holds that $\{A,|\varphi\rangle\}\ C\ \{A,|\varphi\rangle\}$. 
\end{proof}

\begin{proof}[Proof of Rule (Conjunction)] We prove that the rule (Conjunction) is sound for total correctness. If $\sigma\models A_1\wedge A_2$, then $\sigma\models A_1$ and $\sigma\models A_2$. For each $i\in\{1,2\}$, by the assumption $\models_\mathit{tot}\left\{A_i,|\varphi\rangle\right\}\ C\ \left\{B_i,|\psi\rangle\right\}$, if $$(C,\sigma,|\varphi\rangle(\sigma))\rightarrow^\ast(\downarrow,\sigma_i^\prime,|\psi_i^\prime\rangle),$$ then we have $\sigma_i^\prime\models B_i$ and $|\psi\rangle=|\psi_i^\prime\rangle.$ Note that the classical programming language underlining $\mathbf{RQC}^{++}$ is assumed to be deterministic. Then $\mathbf{RQC}^{++}$ is deterministic too. Thus, $\sigma_1^\prime=\sigma_2^\prime$, and $\sigma_1^\prime\models B_1\wedge B_2$. Therefore, it holds that $$\models_\mathit{tot}\left\{A_1\wedge A_2,|\varphi\rangle\right\}\ C\ \left\{B_1\wedge B_2,|\psi\rangle\right\}.$$
\end{proof}

\begin{proof}[Proof of Rule ($\exists$-Introduction)] We prove the soundness of ($\exists$-Introduction) for partial correctness. If $\sigma\models(\exists x)A$ and $$(C,\sigma,|\varphi\rangle(\sigma))\rightarrow^\ast(\downarrow,\sigma^\prime,|\psi^\prime\rangle),$$ then there exists an element $a$ such that $\sigma\models A[a/x]$. Thus, we have $\sigma[a/x]\models A$. Since $x\notin\mathit{var}(C)$, by induction on the length of $C$, it is easy to show that $$(C,\sigma[a/x],|\psi\rangle(\sigma[a/x]))\rightarrow^\ast(\downarrow,\sigma^\prime[a/x],|\psi^\prime\rangle).$$ By the assumption $\models_\mathit{par}\{A,|\varphi\rangle\}\ C\ \{B,|\psi\rangle\}$, we obtain $\sigma^\prime[a/x]\models B$ and $|\psi\rangle=|\psi^\prime\rangle(\sigma^\prime[a/x])$. Furthermore, by the assumption $x\notin\mathit{free}(B)$ and $\mathit{Inv}(|\psi\rangle: x)$, we have $\sigma^\prime\models B$ and $|\psi\rangle=|\psi^\prime\rangle(\sigma)$. This proves $$\models_\mathit{par}\{(\exists x)A,|\varphi\rangle\}\ C\ \{B,|\psi\rangle\}.$$ 
\end{proof}

\begin{proof}[Proof of Rule (Invariance-Con)] We only consider the soundness for partial correctness. If $\sigma\models A\wedge A^\prime$ and $$(C,\sigma,|\varphi\rangle(\sigma))\rightarrow^\ast(\downarrow,\sigma^\prime,|\psi^\prime\rangle),$$ then $\sigma\models A$ and $\sigma\models A^\prime$. By the assumption $\models_\mathit{par}\{A,|\varphi\rangle)\ C\ \{B,|\psi\rangle\}$, we have $\sigma^\prime\models B$ and $|\psi^\prime\rangle=|\psi\rangle(\sigma^\prime)$. Since $\mathit{free}(A^\prime)\cap\mathit{change}(C)=\emptyset$, it holds that $\sigma|\mathit{free}(A^\prime)=\sigma^\prime|\mathit{free}(A^\prime)$. Consequently, $\sigma^\prime\models A^\prime$, and it follows that $\sigma^\prime\models B\wedge A^\prime$. This proves $$\models_\mathit{par}\{A\wedge A^\prime,|\varphi\rangle)\ C\ \{B\wedge A^\prime,|\psi\rangle\}.$$ 
\end{proof}

\begin{proof}[Proof of Rule (Substitution)] If $\sigma\models A[\overline{x}:=\overline{t}]$ and $$(C,\sigma,|\varphi\rangle(\sigma))\rightarrow^\ast(\downarrow,\sigma^\prime,|\psi^\prime\rangle),$$ then by the substitution lemma in the first-order logic we obtain $\sigma[\overline{x}:=\sigma(\overline{t})]\models A$. With the assumption $[\overline{x}\cup\mathit{var}(\overline{t})]\cap\mathit{change}(C)=\emptyset$, it is easy to show that $$(C,\sigma[\overline{x}:=\sigma(\overline{t})],|\varphi\rangle(\sigma[\overline{x}:=\sigma(\overline{t})]))\rightarrow^\ast(\downarrow,\sigma^\prime[\overline{x}:=\sigma^\prime(\overline{t})],|\psi^\prime\rangle)$$ by induction on the length of $C$. Thus, by the assumption $\models_\mathit{par}\{A,|\varphi\rangle)\ C\ \{B,|\psi\rangle\}$, we obtain $\sigma^\prime[\overline{x}:=\sigma^\prime(\overline{t})]\models B$ and $|\psi^\prime\rangle=|\psi\rangle(\sigma^\prime[\overline{x}:=\sigma(\overline{t})])$. Thus, by the substitution lemma, it holds that $\sigma^\prime\models B[\overline{x}:=\overline{t}]$. On the other hand, the assumption $\mathit{Inc}(|\psi\rangle:\overline{x})$ implies $|\psi^\prime\rangle=|\psi\rangle(\sigma^\prime)$. This proves $$\models_\mathit{par}\{A[\overline{x}:=\overline{t}],|\varphi\rangle)\ C\ \{B[\overline{x}:=\overline{t}],|\psi\rangle\}.$$ 
\end{proof}

\end{document}